\PassOptionsToPackage{unicode}{hyperref}
\PassOptionsToPackage{hyphens}{url}
\PassOptionsToPackage{dvipsnames,svgnames,x11names}{xcolor}
\documentclass[
  12pt]{article}

\usepackage{amsmath,amssymb}
\usepackage{iftex}

\usepackage{graphicx,psfrag,epsf}
\usepackage{enumerate}
\usepackage{natbib}
\usepackage{url} 
\usepackage{amssymb}
\usepackage{extarrows} 
\usepackage{mathtools}\usepackage{amsmath}
\usepackage{algorithm}
\usepackage{algorithmic}
\usepackage{amsthm}
\usepackage{subcaption}  
\usepackage{booktabs}  
\usepackage{xcolor,colortbl}
\usepackage{url} 
\usepackage{longtable}
\usepackage{multirow}
\usepackage{array}
\usepackage{rotating}
\usepackage{comment}  
\usepackage{natbib}
\usepackage{comment}

\usepackage{bbm}

\newtheorem{theorem}{Theorem}

\newtheorem{lemma}{Lemma}

\ifPDFTeX
  \usepackage[T1]{fontenc}
  \usepackage[utf8]{inputenc}
  \usepackage{textcomp} 
\else 
  \usepackage{unicode-math}
  \defaultfontfeatures{Scale=MatchLowercase}
  \defaultfontfeatures[\rmfamily]{Ligatures=TeX,Scale=1}
\fi
\usepackage{lmodern}
\ifPDFTeX\else  
\fi
\IfFileExists{upquote.sty}{\usepackage{upquote}}{}
\IfFileExists{microtype.sty}{
  \usepackage[]{microtype}
  \UseMicrotypeSet[protrusion]{basicmath} 
}{}
\makeatletter
\@ifundefined{KOMAClassName}{
  \IfFileExists{parskip.sty}{%
    \usepackage{parskip}
  }{
    \setlength{\parindent}{0pt}
    \setlength{\parskip}{6pt plus 2pt minus 1pt}}
}{
  \KOMAoptions{parskip=half}}
\makeatother
\usepackage{xcolor}
\makeatletter
\ifx\paragraph\undefined\else
  \let\oldparagraph\paragraph
  \renewcommand{\paragraph}{
    \@ifstar
      \xxxParagraphStar
      \xxxParagraphNoStar
  }
  \newcommand{\xxxParagraphStar}[1]{\oldparagraph*{#1}\mbox{}}
  \newcommand{\xxxParagraphNoStar}[1]{\oldparagraph{#1}\mbox{}}
\fi
\ifx\subparagraph\undefined\else
  \let\oldsubparagraph\subparagraph
  \renewcommand{\subparagraph}{
    \@ifstar
      \xxxSubParagraphStar
      \xxxSubParagraphNoStar
  }
  \newcommand{\xxxSubParagraphStar}[1]{\oldsubparagraph*{#1}\mbox{}}
  \newcommand{\xxxSubParagraphNoStar}[1]{\oldsubparagraph{#1}\mbox{}}
\fi
\makeatother

\usepackage{longtable,booktabs,array}
\usepackage{calc} 
\usepackage{etoolbox}
\makeatletter
\patchcmd\longtable{\par}{\if@noskipsec\mbox{}\fi\par}{}{}
\makeatother
\IfFileExists{footnotehyper.sty}{\usepackage{footnotehyper}}{\usepackage{footnote}}
\makesavenoteenv{longtable}
\usepackage{graphicx}
\makeatletter
\def\maxwidth{\ifdim\Gin@nat@width>\linewidth\linewidth\else\Gin@nat@width\fi}
\def\maxheight{\ifdim\Gin@nat@height>\textheight\textheight\else\Gin@nat@height\fi}
\makeatother
\setkeys{Gin}{width=\maxwidth,height=\maxheight,keepaspectratio}
\makeatletter
\def\fps@figure{htbp}
\makeatother

\makeatletter
\@ifpackageloaded{caption}{}{\usepackage{caption}}
\AtBeginDocument{%
\ifdefined\contentsname
  \renewcommand*\contentsname{Table of contents}
\else
  \newcommand\contentsname{Table of contents}
\fi
\ifdefined\listfigurename
  \renewcommand*\listfigurename{List of Figures}
\else
  \newcommand\listfigurename{List of Figures}
\fi
\ifdefined\listtablename
  \renewcommand*\listtablename{List of Tables}
\else
  \newcommand\listtablename{List of Tables}
\fi
\ifdefined\figurename
  \renewcommand*\figurename{Figure}
\else
  \newcommand\figurename{Figure}
\fi
\ifdefined\tablename
  \renewcommand*\tablename{Table}
\else
  \newcommand\tablename{Table}
\fi
}
\@ifpackageloaded{float}{}{\usepackage{float}}
\floatstyle{ruled}
\@ifundefined{c@chapter}{\newfloat{codelisting}{h}{lop}}{\newfloat{codelisting}{h}{lop}[chapter]}
\floatname{codelisting}{Listing}

\makeatother
\makeatletter
\@ifpackageloaded{caption}{}{\usepackage{caption}}
\@ifpackageloaded{subcaption}{}{\usepackage{subcaption}}
\makeatother

\ifLuaTeX
  \usepackage{selnolig}  
\fi
\usepackage[]{natbib}
\usepackage{bookmark}

\IfFileExists{xurl.sty}{\usepackage{xurl}}{} 
\hypersetup{
  pdftitle={Title},
  pdfauthor={Author 1; Author 2},
  pdfkeywords={3 to 6 keywords, that do not appear in the title},
  colorlinks=true,
  linkcolor={blue},
  filecolor={Maroon},
  citecolor={Blue},
  urlcolor={Blue},
  pdfcreator={LaTeX via pandoc}}

\begin{document}

\def\spacingset#1{\renewcommand{\baselinestretch}%
{#1}\small\normalsize} \spacingset{1}


{
  \title{\bf  Mixed Effects Mixture of Experts: Modeling Double Heterogeneous Trajectories}
  \author{Xinkai Yue, Xiaodong Yan, Haohui Han and
    Liya Fu\thanks{
    The authors gratefully acknowledge \textit{the National Natural Science Foundation of China (No. 12371295, 12371292) and the Shaanxi Fundamental Science Research Project for Mathematics and Physics (No. 22JSY002)}}
    \thanks{Email: fuliya@mail.xjtu.edu.cn} \\
    School of Mathematics and Statistics, Xi’an Jiaotong University}
  \maketitle
} 

\bigskip
\begin{abstract}
Linear mixed-effects model (LMM) is a cornerstone of longitudinal data analysis, but is limited to adeptly make heterogeneous analyses predictable under both group-specific fixed effects and subject-specific random effects.  To address this challenge, we propose a novel statistical framework by using a large model prototype: a mixed effects mixture of experts model (MEMoE). This framework integrates the `divide-and-conquer' paradigm of Mixture of Experts Models with classical mixed-effect modeling. In the proposed MEMoE, each `expert' is a full LMM dedicated to capturing the longitudinal trajectory of a specific latent subpopulation, while another model “gating function” learns to route subjects to the most appropriate expert in a data-driven manner based on baseline covariates. We develop a robust inferential procedure for parameter estimation based on the Laplace Expectation-Maximization algorithm, with standard errors calibrated using robust sandwich estimators to account for potential model misspecification.  Extensive simulation studies and an empirical application demonstrate that MEMoE outperforms both traditional single-population LMM and conventional Mixture of Experts models in terms of parameter recovery, classification accuracy, and overall model fit. 
\end{abstract}

\noindent%
\vfill

\newpage
\spacingset{1.8} 

\section{Introduction}\label{sec-intro}

Longitudinal data provides substantial information on dynamic processes, including disease progression, cognitive development, and behavioral changes \citep{diggle2002, fitzmaurice2011}. The linear mixed-effects model (LMM) accounts for subject-level correlations in longitudinal data by adding random effects, avoids the underestimation of standard errors, and thus ensures the validity of hypothesis testing for correlated data \citep{hedaker2006} by separating variance into fixed (population-level) and random (subject-level) components \citep{laird1982,verbeke2000,pinheiro2000,little2019}. To address the limitations of the homogeneous population assumption in longitudinal data analysis, mixture linear mixed models \citep{verbeke1996, muthen2004} or subgroup analysis approaches \citep{yang2019high, yan2021subgroup, huang2023integrative} have been developed to accommodate non-Gaussian random-effect distributions and enable clustering of subjects with distinct trajectories.  Nonetheless, LMM and its extensions have limited capacity to accurately predict under heterogeneous model structures that simultaneously incorporate group-specific fixed effects and subject-specific random effects.

The large model prototype ``mixture of experts'' (MoE) paradigm is a foundational technique that leverages a gating function to route inputs to appropriate specialized submodels, thereby facilitating the adaptive assignment of expert-specific (i.e., group-specific) effects \citep{jacobs1991}. In this framework, each expert focuses on a distinct region of the input space and is accountable for a specific data subset. This methodology has demonstrated efficacy across various domains, enabling flexible capture and discovery of latent subgroups in regression data \citep{yuksel2012, bishop2006}. In large-scale data contexts, MoE enhances scalability—a key reason for its adoption in modern deep learning—by efficiently managing complex, multimodal distributions without imposing uniformity assumptions \citep{shazeer2017a}.
Recent applications to longitudinal trajectories further illustrate MoE's capacity to identify latent classes while accounting for heterogeneity, as in growth curve analysis \citep{gao2002, quiroz2013}. The adaptive gating mechanism weights expert contributions differentially across individuals or time points, thereby effectively accommodating population heterogeneity \citep{he2025}. 
Consequently, MoE is particularly advantageous for longitudinal applications involving divergent subgroup patterns, such as varying treatment responses in patient cohorts or differential growth curves in educational studies, where it can uncover time-dependent latent classes, such as heterogeneous brain activity trajectories in infant emotional reactivity research \citep{che2023}.

However, standard MoE models lack an explicit mechanism to account for within correlation, simply treating all observations as independent, which can lead to biased estimates in longitudinal settings \citep{jordan1994}. Hierarchical extensions of MoE have been proposed, but they typically fail to incorporate random effects within the experts, thereby limiting their ability to capture unobserved subject-level heterogeneity \citep{xu1996}. In contrast, existing mixed-effects extensions of regression models, such as mixtures of linear mixed models, effectively capture correlations yet fail to leverage the gating-expert architecture of MoE. Recent research has sought to bridge this gap. \citet{fung2022} proposed a mixed MoE for multilevel data, demonstrating that such models can approximate arbitrary mixed-effects distributions. Similarly, \citet{kock2025deep} developed a deep mixture of linear mixed models to handle irregular longitudinal data with complex temporal dynamics, incorporating deep latent factors to model high-dimensional random effects. These pioneering studies underscore the promise of combining mixtures and random effects. However, practical parameter estimation remains difficult, stemming not only from the presence of hidden variables but also from the substantial computational burden associated with integrating over random effects.

To address the challenges of classical LMMs and MoE models, we propose a Mixed-Effects Mixture of Experts (MEMoE) model, which unifies the MoE framework with subject-specific random effects. In MEMoE, each expert comprises a linear mixed-effects model with its own random intercepts and slopes, governed by a multivariate normal prior; the gating function then probabilistically assigns observations to the respective experts. This approach captures heterogeneity between-subjects through mixture components while simultaneously accounting for within correlations through random effects, thereby bridging multilevel data structures with MoE architecture. The key contributions are listed as follows:

(i) For methodology, the proposed MEMoE model advances the synergy between MoE and LMMs, because it flexibly models double heterogeneous longitudinal trajectories, accommodating subject-specific variability.  Relative to mixtures of linear mixed models, MEMoE employs explicit gating for subgroup discovery; in contrast to standard MoE models, it formally incorporates within correlations.

(ii) In computation, we develop a new Laplace-EM algorithm to address the challenge of intractable marginal likelihood, which integrates over the latent random effects. 

(iii) Practically, MEMoE is especially appropriate for applications such as patient trajectory modeling in healthcare, where latent subgroups may exhibit distinct progression patterns (e.g., heterogeneous epidemic disease progression, \citep{cui2022}); educational growth curves accounting for student-specific differences; or socioeconomic trends with subgroup-specific dynamics.

The remainder of this paper is organized as follows. In Section 2, we present the MEMoE models and the Laplace-EM algorithm. Section 3 provides a set of constructions for response prediction. In Section 4, we establish the theoretical properties of the proposed estimator, including the consistency of the Laplace-EM estimator and the asymptotic normality of the predictor. Section 5 evaluates the performance of the proposed method through extensive simulation studies and real-data analysis. Section 6 gives some concluding remarks.

\section{Methods}\label{sec-meth}
This section introduces the MEMoE framework, which demonstrates the primary advantage of making both observation- and subgroup-specific heterogeneous analyses predictable when the subgroup structure is unknown.

\subsection{The  Mixed Effects Mixture of Experts Models}
\label{ssec:MEMoE}

Consider the data collected from $N$ subjects.
Let $y_{ij}$ denote the  $j$-th outcome for subject $i$, and ${x}_{ij}=(1, x_{ij1},\dots,x_{ij(p-1)})^\top\in\mathbb{R}^p$ be the fixed-effect covariates and
${z}_{ij}\in\mathbb{R}^q$ be the random-effect covariates, where $i=1,\ldots N$ and  $j=1,\dots,n_i$.
We posit $K$ latent subgroups and assign to each observation a latent label $v_{ij} \in \{1, \dots, K\}$ indicating its subgroup membership. Given the expert label $v_{ij}=k$, the response follows an expert-specific mixed effects model: 
\begin{equation*}\label{eq:lmm-obs}
y_{ij} \;=\; x_{ij}^\top \beta_k \;+\; z_{ij}^\top u_i \;+\; \varepsilon_{ij}^k,
\end{equation*}
where ${\beta}_{k}\in\mathbb{R}^p$ is an expert-specific fixed effect,  ${u}_i\in\mathbb{R}^q$ is a subject-level random effect shared between experts, and the random error $\varepsilon_{ij}^k\sim\mathcal N(0,\sigma_k^2)$. Then
\begin{equation}\label{eq:expert-cond}
y_{ij}\mid x_{ij},z_{ij},u_i, v_{ij}=k \;\sim\;
\mathcal N\!\big(x_{ij}^\top \beta_k + z_{ij}^\top u_i,\ \sigma_k^2\big).
\end{equation}

To identify which expert a subject belongs to, we assume that the gating function depends only on the covariates at the subject-level:
\begin{equation*}\label{eq:gate-cat}
v_{ij}\mid {x}_{ij},{z}_{ij}\ \sim\
\mathrm{Multinomial}\big(\,\pi_1({x}_{ij},{z}_{ij};\alpha),\ldots,\pi_K({x}_{ij},{z}_{ij};\alpha)\big),
\end{equation*}
where $\alpha$ is an unknown parameter vector.
The covariate-dependent weight $\pi_k({x}_{ij}, {z}_{ij};\alpha) $ can be interpreted as the probability that the \((i,j)\)-th observation comes from the  $k$th expert model and serves to automatically discover latent subgroups and regime changes. 

To provide a more explanatory and predictive structure for 
subject-level random effects $u_i$, we construct the following covariate-dependent model:
\begin{equation}\label{eq:random-ui}
u_i \mid w_i \sim \mathcal{N}(\kappa  w_i, \Sigma),
\end{equation}
where $w_i$ collects subject-level covariates (e.g., demographics, baseline measures),  $\kappa \in \mathbb{R}^{q \times d}$, and $\Sigma$ is a $q \times q$  covariance matrix. This specification encodes explainable between-subject variation through $w_i$; setting $\kappa = 0$ recovers the commonly used random-effects  $u_i \sim \mathcal{N}(0, \Sigma)$. In this zero-mean random-effects case, we refer to the resulting expert model as the random-effects mixture-of-experts (ReMoE) specification.
Figure \ref{MEMoEfigure} shows the workflow framework of the proposed MEMoE model.

\begin{figure}[htbp]
\centering
\includegraphics[width=1\textwidth]{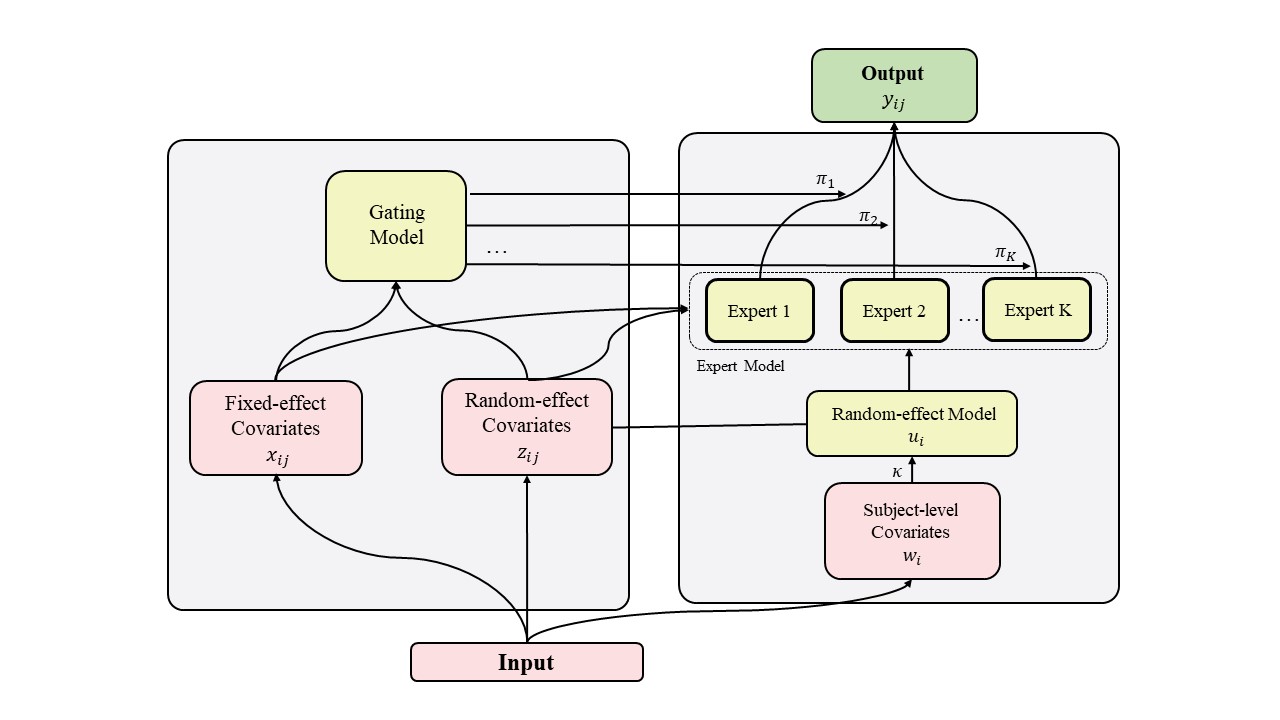}
\caption{Structure of the MEMoE. Pink blocks represent data inputs, yellow blocks denote the created model framework, and the green block signifies the model output.}
\label{MEMoEfigure}
\end{figure}

The proposed MEMoE specification is a unifying formulation that nests several widely used expert models. In particular, when each subject has a single observation ($n_i=1$ for all $i$) or the random effects vanish (i.e., $\Sigma=0$, such that $u_i\equiv 0$), the model coincides with the classical MoE regression, dedicating an expert model to each partitioned observation. When the gating function is degenerate with only one expert, i.e., $\beta_1=\cdots=\beta_K$ and $\sigma_1^2=\cdots=\sigma_K^2$,  MEMoE reduces to a single homoscedastic linear mixed effects model.

\subsection{Laplace-EM Algorithm}
\label{ssec:exact-likelihood}
In this subsection, we construct a Laplace EM algorithm to estimate the parameters in the MEMoE.  The E-step relies on a Laplace approximation to compute the approximate modes of the random effects \citep{Breslow1993, skaug2006}. The M-step updates model parameters via block-wise optimizations. To ensure stable and monotonic convergence, particularly for non-convex components, we integrate a majorize-minimize strategy \citep{HunterLange2004MM}.

Conditional on the random effect $u_i$ and covariates $x_{ij}$ and $z_{ij}$, the conditional density of $y_{ij}$ is
\begin{equation*}\label{eq:cond-mixture-ij}
f\!\left(y_{ij}\mid u_i,x_{ij},z_{ij};\Psi\right)
=\sum_{k=1}^K \pi_k(x_{ij},z_{ij};\alpha)\;
\varphi\!\left(y_{ij};\,x_{ij}^\top\beta_k+z_{ij}^\top u_i,\ \sigma_k^2\right), 
\end{equation*}
where $\varphi(\cdot)$ is a normal density 
$$ \varphi\!\left(y_{ij};\,x_{ij}^\top\beta_k+z_{ij}^\top u_i,\ \sigma_k^2\right) = \frac{1}{\sqrt{2\pi\sigma_k^2}} \exp\left\{ -\frac{\left(y_{ij} - (x_{ij}^\top\beta_k + z_{ij}^\top u_i)\right)^2}{2\sigma_k^2} \right\},
$$
and $\Psi$ represents the set of all model parameters and belongs to a parameter space $\Theta$. Under the assumption  that $\{y_{ij}\}_{j=1}^{n_i}$ are conditionally independent given $u_i$, the likelihood function for the repeated measurements takes the following form:
\begin{align}\label{eq:cond-mix-i}
L(\Psi) \nonumber
&= \prod_{i=1}^N f(y_{i1}, y_{i2}, \ldots, y_{in_i} \mid \{x_{ij}, z_{ij}\}_{j=1}^{n_i}; \Psi) \\ \nonumber
&= \prod_{i=1}^N \int \prod_{j=1}^{n_i} f(y_{ij} \mid u_i, x_{ij}, z_{ij}; \Psi) \,
\varphi(u_i; \kappa w_i, \Sigma) \,
du_i \\ 
&= \prod_{i=1}^N \int \left[ \prod_{j=1}^{n_i} \sum_{k=1}^K \pi_k(x_{ij}, z_{ij}; \alpha) \, \varphi\big(y_{ij}; x_{ij}^\top\beta_k + z_{ij}^\top u_i, \sigma_k^2\big) \right] \varphi\big(u_i; \kappa w_i, \Sigma\big) \, du_i.
\end{align}
This integral has no closed-form in general because the subject-level mixture is embedded within the product over $j$, and all measurements share the same random effect $u_i$. Expanding the product of sums leads to a mixture of $K^{n_i}$ multivariate Gaussians. We therefore approximate the integral in \eqref{eq:cond-mix-i} using the Laplace method \citep{Breslow1993, Tierney1986}. Let $h_i(u_i; \Psi)$ denote the logarithm of the integrand:
$$
h_i(u_i; \Psi) = \log \varphi(u_i; \kappa w_i, \Sigma) + \sum_{j=1}^{n_i} \log \left[ \sum_{k=1}^{K} \pi_k(x_{ij}, z_{ij}; \alpha) \varphi\left(y_{ij}; x_{ij}^\top\beta_k + z_{ij}^\top u_i, \sigma_k^2\right) \right].
$$
The Laplace method approximates the integral by forming a second-order Taylor expansion of $h_i(u_i; \Psi)$ around its mode, $\hat{u}_i = \arg\max_{u_i} h_i(u_i; \Psi)$. Let ${H}_i={H}_i(\hat{u}_i; \Psi) = - \nabla^2_{u_i}h_i(u_i; \Psi) \big|_{u_i = \hat{u}_i}$ be the negative Hessian matrix evaluated in the mode, then
\begin{align*}
H_i &=  \Sigma^{-1} - \sum_{j=1}^{n_i} \sum_{k=1}^K \frac{\pi_k(x_{ij}, z_{ij}; \alpha) \varphi_{ijk}(\hat{u}_i)}{\sum_{\ell=1}^K \pi_\ell(x_{ij}, z_{ij}; \alpha) \varphi_{ij\ell}(\hat{u}_i)} \left(   \frac{e_{ijk}(\hat{u}_i)^2}{\sigma_k^4} -\frac{1}{\sigma_k^2}\right) z_{ij}z_{ij}^\top \\
& + \sum_{j=1}^{n_i} \left( \sum_{k=1}^K \frac{\pi_k(x_{ij}, z_{ij}; \alpha) \varphi_{ijk}(\hat{u}_i)}{\sum_{\ell=1}^K \pi_\ell(x_{ij}, z_{ij}; \alpha) \varphi_{ij\ell}(\hat{u}_i)} \frac{e_{ijk}(\hat{u}_i)}{\sigma_k^2} z_{ij} \right) \left( \sum_{k=1}^K \frac{\pi_k(x_{ij}, z_{ij}; \alpha) \varphi_{ijk}(\hat{u}_i)}{\sum_{\ell=1}^K \pi_\ell(x_{ij}, z_{ij}; \alpha) \varphi_{ij\ell}(\hat{u}_i)} \frac{e_{ijk}(\hat{u}_i)}{\sigma_k^2} z_{ij}^\top \right),
\end{align*}
where $e_{ijk}(\hat u_i) = y_{ij} - x_{ij}^\top\beta_k - z_{ij}^\top\hat u_i$. Therefore,
\begin{eqnarray*}
H_i &\simeq &  \Sigma^{-1} + \sum_{j=1}^{n_i} \sum_{k=1}^K \frac{\pi_k(x_{ij}, z_{ij}; \alpha) \varphi_{ijk}(\hat{u}_i)}{\sum_{\ell=1}^K \pi_\ell(x_{ij}, z_{ij}; \alpha) \varphi_{ij\ell}(\hat{u}_i)} \frac{1}{\sigma_k^2} z_{ij}z_{ij}^\top
\end{eqnarray*}
is positive definite. The Laplace approximation to the likelihood function is then given by:
$$
L( \Psi) \approx \prod_{i=1}^N \left[ (2\pi)^{q/2} |{H}_i|^{-1/2} \exp\left\{ h_i(\hat{u}_i; \Psi) \right\} \right]. 
$$
Taking the logarithms yields the approximated log-likelihood for the parameter set $\Psi$:
\begin{equation*}
\ell_{\rm LA}(\Psi) = \frac{N q}{2}\log(2\pi)+\sum_{i=1}^N \left[ h_i(\hat{u}_i; \Psi) - \frac{1}{2} \log|{H}_i| \right]. 
\end{equation*}

Optimizing the Laplace-approximated log-likelihood $\ell_{{\rm LA}}(\Psi)$ poses significant challenges.
In this context, we implement an EM-inspired minorize-maximize algorithm that iteratively builds and maximizes a lower-bound surrogate of $\ell_{\rm LA}(\Psi)$, thereby ensuring steady progress toward a local optimum. 
Given the current iterate $\Psi^{(t)}$, we construct a surrogate lower bound $Q_{{\rm LA}}\big(\Psi \mid \Psi^{(t)}\big)$ for $\ell_{{\rm LA}}(\Psi)$ and maximize it over $\Psi$, thereby increasing $\ell_{{\rm LA}}(\Psi) - \ell_{{\rm LA}}(\Psi^{(t)})$. Concretely, $Q_{{\rm LA}}$ is obtained by a Laplace expansion of the subject-level integral at the posterior mode $\hat{u}_i$ computed by the negative Hessian $H_i$.

\textbf{E-step (random effects mode and responsibilities).}

Let $\hat u_i^{(t)} = \arg\max_{u_i}\, h_i(u_i;\Psi^{(t)}).$
Denote $\gamma_{ijk}^{(t)} $ as the probability that observation is generated by expert $k$, conditional on the latent random effect:
\begin{equation}\label{eq:resp}
\gamma_{ijk}^{(t)} = 
\frac{\pi_k(x_{ij},z_{ij};\alpha^{(t)})\,
\varphi\!\big(y_{ij};\,x_{ij}^\top\beta_k^{(t)}+z_{ij}^\top \hat u_i^{(t)},\,(\sigma^2_k)^{(t)}\big)}
{\displaystyle \sum_{\ell=1}^{K}\pi_\ell(x_{ij},z_{ij};\alpha^{(t)})\,
\varphi\!\big(y_{ij};\,x_{ij}^\top\beta_\ell^{(t)}+z_{ij}^\top \hat u_i^{(t)},\,(\sigma^2_{\ell})^{(t)}\big)}.
\end{equation}
It is called the responsibility, as it quantifies how `responsibility' expert $k$ is for $y_{ij}$ at the current conditional state.

Form a touching lower limit $Q_{{\rm LA}}(\Psi\mid\Psi^{(t)})$  for
$\ell_{{\rm LA}}(\Psi)$ using Jensen’s inequality to the subject-level term $\log \left[ \sum_{k=1}^{K} \pi_k(x_{ij}, z_{ij}; \alpha) \varphi\left(y_{ij}; x_{ij}^\top\beta_k + z_{ij}^\top u_i, \sigma_k^2\right) \right]$  evaluated at $\hat u_i^{(t)}$, 
and linearize  $-\log|H_i|$ at the current estimate $H_i^{(t)}$, with $-\log|H_i|
\;\ge\;
-\log|H_i^{(t)}|
-\,\mathrm{tr}\Big\{\big(H_i^{(t)}\big)^{-1}\big(H_i-H_i^{(t)}\big)\Big\}$. 
Then the $Q_{{\rm LA}}\big(\Psi \mid \Psi^{(t)}\big)$ is
\begin{eqnarray}\label{eq:Q-function}\nonumber
Q_{{\rm LA}}\big(\Psi \mid \Psi^{(t)}\big)
&=& \sum_{i=1}^N\Bigg\{
\log \varphi\big(\hat u_i^{(t)};\kappa w_i,\Sigma\big)
+ \sum_{j=1}^{n_i}\sum_{k=1}^K \gamma_{ijk}^{(t)}
\Big[
\log \pi_k(x_{ij},z_{ij};\alpha)\\
&+&  \log \varphi\big(y_{ij}; x_{ij}^\top\beta_k+z_{ij}^\top \hat u_i^{(t)}, \sigma_k^2\big)
-\log\gamma_{ijk}^{(t)}
\Big]\Bigg\} {- \frac{1}{2}\sum_{i=1}^N \mathrm{tr}\big((H_i^{(t)})^{-1} H_i\big)}.
\end{eqnarray}


\textbf{M-step: parameter updates to maximize the objective function in \eqref{eq:Q-function}.}

Gating parameter vector $\alpha$ is updated by the following formula:
\begin{equation*}\label{eq:update-alpha}
\alpha^{(t+1)}
=\arg\max_{\alpha}\
\sum_{i=1}^N \sum_{j=1}^{n_i} \sum_{k=1}^{K}
\gamma_{ijk}^{(t)} \, \log \pi_k(x_{ij},z_{ij};\alpha).
\end{equation*}
Expert's coefficients $\{\beta_k\}_{k=1}^K$:
\begin{equation*}\label{eq:update-beta}
\hat\beta_k^{(t+1)}=\Big(\sum_{i,j}\tfrac{\gamma_{ijk}^{(t)}}{\hat\sigma_k^{2,(t)}}x_{ij}x_{ij}^\top\Big)^{-1}
\Big(\sum_{i,j}\tfrac{\gamma_{ijk}^{(t)}}{\hat\sigma_k^{2,(t)}} x_{ij}(y_{ij}-z_{ij}^\top \hat u_i^{(t)})\Big).
\end{equation*}
Noise variances $\{\sigma_k^2\}_{k=1}^K$:
\begin{equation*}\label{eq:update-sigma}
(\sigma^2_k)^{(t+1)} =
\frac{\displaystyle \sum_{i=1}^N \sum_{j=1}^{n_i} \gamma_{ijk}^{(t)}
\Big[\big(y_{ij}-x_{ij}^\top \beta_k^{(t+1)} - z_{ij}^\top \hat u_i^{(t)}\big)^2
+ z_{ij}^\top( H_i^{(t)})^{-1} z_{ij}\Big]}
{\displaystyle \sum_{i=1}^N \sum_{j=1}^{n_i} \gamma_{ijk}^{(t)}}.
\end{equation*}
Random-effect mean parameter matrix $\kappa$:
\begin{equation*}\label{eq:update-kappa}
\kappa^{(t+1)}
= \Big(\sum_{i=1}^N \hat u_i^{(t)}w_i^\top\Big)\Big(\sum_{i=1}^N w_iw_i^\top\Big)^{-1}.
\end{equation*}
Random-effect covariance matrix $\Sigma$:
\begin{equation*}\label{eq:update-sigmaRE}
\Sigma^{(t+1)} =
\frac{1}{N} \sum_{i=1}^{N}
\Big[
\big(\hat u_i^{(t)} - \kappa^{(t+1)} w_i\big)
\big(\hat u_i^{(t)} - \kappa^{(t+1)} w_i\big)^\top
+ (H_i^{(t)})^{-1}
\Big].
\end{equation*}

The algorithm alternates between computing subject-level membership probabilities and updating the model parameters. At iteration $t$, the E-step computes the responsibilities ${\gamma_{ijk}^{(t)}}$, $\hat u_i^{(t)}$, and $H_i^{(t)}$, using the current parameter vector $\Psi^{(t)}$. 
Subsequently, the M-step updates the elements of $\Psi$ via the specified iterative formulas.

\section{Prediction Sets}
\label{sec:conc}
Quantifying predictive uncertainty is essential for reliable inference in longitudinal studies that exhibit both population heterogeneity and within-subject dependence. We develop prediction sets for the MEMoE model — a covariate-gated mixture framework that captures the latent subgroup structures while accounting for subject-specific random effects. Given a fitted model based on $N$ subjects, a target confidence level $1-q\in(0,1)$, and a tuple of new covariates $(x_{\text{new}}, z_{\text{new}}, w_{\text{new}})$, our aim is to construct a conditionally valid prediction set $\Omega(x_{\text{new}}, z_{\text{new}}, w_{\text{new}})$ for a future response $y_{\text{new}}$, such that:
\begin{equation}\label{eq:p}
\Pr\left( y_{\text{new}} \in \Omega(x_{\text{new}}, z_{\text{new}}, w_{\text{new}}) \mid x_{\text{new}}, z_{\text{new}}, w_{\text{new}} \right) \geq 1 - q,
\end{equation}
a guaranty targeted at fixed covariates rather than averaged over their distribution (cf. conditional versus marginal coverage in predictive inference). In words, a prediction set is a range — possibly a union of disjoint intervals — within which a future response is expected to fall with a probability of at least $1-q$; shorter prediction sets are preferred for informativeness, whereas the trivial set $(-\infty,\infty)$ achieves $100\%$ coverage yet conveys no meaningful information. This conditional coverage target mirrors recent developments in predictive inference for MoE models and highlights the advantage of parametric modeling of the conditional law when per-covariate guarantees are required.  

Given the new input, the density of the response $y_{\text{new}}$ can be obtained via
\begin{align*}
f(y_{\text{new}} \mid x_{\text{new}}, z_{\text{new}}, w_{\text{new}}; \Psi)&= \int f(y_{\text{new}} \mid u_{\text{new}}, x_{\text{new}}, z_{\text{new}}; \Psi) \,
\varphi(u_{\text{new}}; \kappa w_{\text{new}}, \Sigma) \, du\\ 
&= \sum_{k=1}^K \pi_k(x_{\text{new}}, z_{\text{new}}; \alpha)  \mathcal{N}\left( y_{\text{new}}; x_{\text{new}}^\top \beta_k + z_{\text{new}}^\top \kappa w_{\text{new}}, \, \sigma_k^2 + z_{\text{new}}^\top \Sigma z_{\text{new}} \right).
\end{align*}
If $v_{\text{new}}=k$, 
we obtain the predictor $\hat{\Gamma}_k= x_{\text{new}}^\top \hat{\beta}_k + z_{\text{new}}^\top \hat{\kappa} w_{\text{new}} $  based on Laplace-EM estimates
and   approximate the variance of   $\hat{\Gamma}_k$ by a working variance matrix $\hat{V}_k = \hat{V}_k^{(\beta)} + \hat{V}_k^{(\kappa)}$, where  $ \hat{V}_k^{(\beta)}  = x_{\text{new}}^\top \left( \sum_{i,j} \frac{\gamma_{ijk}}{\hat{\sigma}_k^2} x_{ij}x_{ij}^\top \right)^{-1} x_{\text{new}}$ and 
$ \hat{V}_k^{(\kappa)} = (w_{\text{new}} \otimes z_{\text{new}})^\top \left( \hat{\Sigma} \otimes \left( \sum_{i=1}^N w_i w_i^\top \right)^{-1} \right)(w_{\text{new}} \otimes z_{\text{new}}).$
Then the variance estimate of the predicted response is
\begin{equation}
\hat{b}_k^2 = z_{\text{new}}^\top \hat{\Sigma} z_{\text{new}} + \hat{\sigma}_k^2 + \hat{V}_k.
\end{equation}

Recall that our goal is to construct a $100(1-q)\%$  prediction set $\Omega_q(x_{\mathrm{new}},z_{\mathrm{new}},w_{\mathrm{new}})$  satisfying \eqref{eq:p}. In MEMoE, the predictive distribution for $y_{\mathrm{new}}$ given covariates $(x_{\mathrm{new}}, z_{\mathrm{new}}, w_{\mathrm{new}})$  is a finite mixture of Gaussian, with each component corresponding to an expert indexed by $k = 1, \dots, K$.
Motivated by the mixture structure, we  establish a prediction set given as follows:
$$
\Omega_q(x_{\mathrm{new}},z_{\mathrm{new}},w_{\mathrm{new}})
\;=\;
\bigcup_{k=1}^{K} [\ell_k,u_k],
$$
where each $[\ell_k,u_k]$ is a estimator $\hat{\Gamma}_k$.
Treating the length of $\Omega_q(x_{\mathrm{new}},z_{\mathrm{new}},w_{\mathrm{new}})$ as 
our budget, it is intuitive that we should allocate more of this budget to the mixture components to which the new predictor is more likely to be assigned--that is, to those groups $k$ with larger value of ${\pi}(x_{\mathrm{new}}, z_{\mathrm{new}})$. To this end, we construct the marginal predictive density $\hat{f}(y)$ using a weighted mixture of Gaussian mixture densities:
$$
\hat{f}(y) = \sum_{k=1}^K \hat{\pi}_k  \frac{1}{\hat{b}_k}  \varphi\left( \frac{y - \hat{\Gamma}_k}{\hat{b}_k} \right),
$$
where $\varphi(\cdot)$ denotes the density function of the standard normal distribution,
and $\hat{\pi}_k = \pi_k(x_{\mathrm{new}}, z_{\mathrm{new}}; \hat{\alpha})$.

Our objective is to construct the shortest set $\Omega_q\subset\mathbb R$ such that $\int_{\Omega_q}\hat f(y)dy\ge 1-q$.
To ensure that the numerical search is limited to a finite and relevant region, we first define a conservative bounding interval $\mathcal Q$ guaranteed to contain the target set. Let $c_q=\Phi^{-1}(1-q/2)$ and define
$$
\mathcal Q=\Big[\min_k\{\hat\Gamma_k-c_q\hat{b}_k\},\ \max_k\{\hat\Gamma_k+c_q\hat{b}_k\}\Big],
$$
which  has a  probability of at least $1-q$. 
Next, we partition $\mathcal Q$  into $M$ equal width subintervals  $Q_r = [a_r, e_r)$  with length $\delta=$ Len$(\mathcal Q)/M$, where $a_r = \min_k\{\hat\Gamma_k-c_q\hat{b}_k\} + (r-1) \delta$ and $ e_r = \min_k\{\hat\Gamma_k-c_q\hat{b}_k\} + r \delta $, for $r = 1, \cdots, M$. Let $y_r=(a_r+e_r)/2 $ be the midpoint. For each cell, we compute the exact probability  under the estimated predictive density $\hat{f}(\cdot)$
$$ 
m_r=\int_{Q_r}\hat f(y)\,dy
=\sum_{k=1}^K \hat\pi_k\!\left[
\Phi\!\Big(\tfrac{e_r-\hat\Gamma_k}{\hat{b}_k}\Big)-\Phi\!\Big(\tfrac{a_r-\hat\Gamma_k}{\hat{b}_k}\Big)\right],
$$
and record the midpoint value $h_r=\hat f(y_r)$.
By sorting $\{h_r\}$ in decreasing order: $h_{(1)}\ge\cdots\ge h_{(M)}$ and applying the same permutation to $\{m_r\}$, we choose the smallest $n^\star$ such that:
$$
n^\star = \min\left\{ n :\sum_{i=1}^{n} m_{(i)}\ \ge\ 1-q\right\}.
$$
Define the selected set
$
\mathcal S =\{ (1),(2),\dots,(n^\star)\}.
$
Form the discrete region on this grid by taking the union of the selected cells
$$
\widehat{\Omega}_q
\;:=\;
\bigcup_{r\in\mathcal S} Q_r
\;=\;
\bigcup_{i=1}^{n^\star} Q_{(i)}.
$$
By construction,
$\int_{\widehat{\Omega}_q} \hat f(y)\,dy=\sum_{i=1}^{n^\star} m_{(i)}\;\ge\;1-q$, and thus
 $\widehat{\Omega}_q$ is a finite grid approximation to the shortest set $\Omega_q$ with $\hat f$–mass at least $1-q$. Algorithm \ref{alg:hdr-grid-concise} presents the detailed steps for constructing a $(1\!-\!q)$ prediction set for MEMoE.

\begin{algorithm}[H]
\caption{Constructing a $(1\!-\!q)$ Prediction Set for MEMoE }
\label{alg:hdr-grid-concise}
\begin{algorithmic}[1]
\STATE \textbf{Inputs:} Confidence level $1-q$;  discretization scale $\delta$, fitted value $\hat\Psi$; covariates $(x_{\mathrm{new}},z_{\mathrm{new}},w_{\mathrm{new}})$.
\STATE \textbf{Form a weighted mixture of Gaussian densities:} 
$$
\hat{f}(y) = \sum_{k=1}^K \hat{\pi}_k(x_{\mathrm{new}}, z_{\mathrm{new}}; \hat{\alpha})  \frac{1}{\hat{b}_k}  \varphi\left( \frac{y - \hat{\Gamma}_k}{\hat{b}_k} \right).
$$
\STATE \textbf{Coverage-safe truncation:} Let $c_q=\Phi^{-1}(1-q/2)$ and set  $\mathcal Q$ that $\int_{\mathcal Q}\hat f(y)dy\ge 1-q$, i.e.
\[
\mathcal Q=[\min_k(\hat\Gamma_k-c_q\hat{b}_k),\max_k(\hat\Gamma_k+c_q\hat{b}_k)].
\]

\STATE \textbf{Grid by density threshold:} Divide $\mathcal Q$ into mesh $Q_r$ of size $\delta$ with midpoints $y_r$ and 
$$
m_r=\int_{Q_r}\hat f(y)\,dy=\sum_{k=1}^K \hat\pi_k\!\left[\Phi\!\Big(\frac{e_r-\hat\Gamma_k}{\hat{b}_k}\Big)-\Phi\!\Big(\frac{a_r-\hat\Gamma_k}{\hat{b}_k}\Big)\right].
$$
Sort $h_r=\hat f(y_r)$ decreasing and find the smallest $n^\star$ such that $$\sum_{i=1}^{n^\star} m_{(i)}\ge 1-q.$$ 
Set the prediction set as:
$$
\hat\Omega_q(x_{\mathrm{new}},z_{\mathrm{new}},w_{\mathrm{new}})= \bigcup_{i=1}^{n^\star} Q_{(i)}.
$$
\STATE \textbf{Output:}
${\Omega}_q(x_{\mathrm{new}},z_{\mathrm{new}},w_{\mathrm{new}})$
\end{algorithmic}
\end{algorithm}

\section{Asymptotical Properties}
In this section, we will prove the consistency of the Laplace-EM estimator $\hat\Psi$ and asymptotic normality of the regression parameter estimator $\hat\beta_k$. Before giving the consistency of the estimator $\hat\Psi$, the following regularity conditions are required.

\begin{itemize}

\item[(A1)]  For covariates, $\mathbb{E}(\|x_{ij}\|^6) < \infty$,  $\mathbb{E}(\|z_{ij}\|^6) < \infty$, and $\mathbb{E}(\|w_i\|^6) < \infty$, where $\|\cdot\|$ denotes $L_2$ norm.

\item[(A2)] The parameter space $\Theta$ 
is a compact set.

\item[(A3)] There exist constants $0 < c_1 \le c_2 < \infty$ such that the eigenvalues of $\Sigma$ satisfy $c_1 \le \lambda_{\min}(\Sigma) \le \lambda_{\max}(\Sigma) \le c_2$. 

\item[(A4)] $ \mathbb{E}(\ell_{{\rm LA}}(\Psi))$ is uniquely maximized at the true parameter  value  $\Psi$.

\item[(A5)] Let $[K]$ denotes the set $\{1,\cdots, k\}$.  Assume that   $\{\alpha_k + c: k \in [K]\}$ corresponds to the same $\pi_k(x_{ij},z_{ij})$ as $\{\alpha_k:k \in [K]\}$ for a constant $c$.

\end{itemize}

Conditions (A1)--(A4) are standard regularity assumptions in the literature on mixture models and linear mixed effect models \citep {Kiefer1959, hennig2000}. Condition (A5) ensures that the  parameters in the gating function are identifiable. Based on these regularity conditions, we establish the consistency of the parameter estimators.

\begin{theorem}
\label{lem:la-consistency}
Let $\hat\Psi$ denote the estimator obtained by maximizing $\ell_{{\rm LA}}(\Psi)$ over the parameter space.
Under the regularity conditions (A1)–(A5), $\hat\Psi$ is a consistent estimator of  $\Psi$ as  $N\to\infty$; that is:
$\hat\Psi \xrightarrow{p} \Psi.$
\end{theorem}

Theorem \ref{lem:la-consistency} establishes the consistency of the global parameter estimator $\hat{\Psi}$ as $N \to \infty$. To establish the asymptotic normality of the regression estimator, we require the following additional regularity conditions.

\begin{itemize}
\item[(B1)] There exists a constant $c_z > 0$ such that the second moment of the random-effect covariates $z_{ij}$ satisfies $ c_z\leq \mathbb{E}[z_{ij}z_{ij}^{\top}] < \infty$.

\item[(B2)] The information matrix is positive definite:
$$
0\prec\mathbb{E}\left[ - \frac{\partial^2 \ell_{\rm LA }}{\partial \beta_k \partial \beta_k^\top} \right] \prec \infty.
$$

\item[(B3)] The function $h_i(u_i; \Psi)$ is strictly concave in a neighborhood of the true random effect $u_i$. Furthermore, the negative Hessian matrix ${H}_i(u_i,\Psi) = - \nabla^2_{u_i}h_i(u_i; \Psi) $  is positive definite. The smallest eigenvalue of  $ - \nabla^2_{u_i} h_i(\hat{u}_i;\Psi) $ satisfies a linear growth condition:
$$
\lambda_{\min}\left( - \nabla^2_{u_i} h_i(\hat{u}_i;\Psi) \right) \ge \lambda_0 + c_1 n_i,
$$
where $\lambda_0 > 0$ depends on the variance of the random effects and $c_1 > 0$ depends on the error variance lower bound.

\item[(B4)] There exists a constant $L > 0$ such that for subjects $i$, 
all indices $j,k,l \in \{1,\ldots,q\}$,
$$
\left|\frac{\partial^3 h_i(u_i; \Psi)}{\partial u_j \partial u_k \partial u_l}\right| 
  \le L \cdot n_i.
$$

\end{itemize}

Condition (B1) ensures that the asymptotic variance of the estimators is well-defined and finite.
Condition (B2) guarantees the uniqueness of $\hat{u}_i$, as stated in \cite{pinheiro2000}. Condition (B3) serves as a regularity condition for deriving the convergence rate of the Taylor expansion, while also ensuring the existence of a unique mode $\hat{u}_i$ in a neighborhood of $u_i$. Condition (B4) assumes that the norm of the third derivative of  $h_i(u_i; \Psi)$ is bounded by a term linear in $n_i$.

\begin{theorem}
\label{thm:an-stud}
Suppose that the regularity conditions (A1)--(A5) and (B1)--(B3) hold. 
For any fixed expert component $k \in [K]$, 
the estimator converges in distribution to a normal distribution as $N \to \infty$:
$$
\sqrt{N}\,(\hat\beta_k-\beta_{k})
\ \xrightarrow{d}\
\mathcal N\big(0,\ V_{\beta_k}\big),
\qquad
V_{\beta_k}=J_{\beta_k}^{-1}K_{\beta_k}J_{\beta_k}^{-1},
$$
where
$$
J_{\beta_k}:=\mathbb{E}\left[ - \frac{\partial^2 \ell_{\rm LA }}{\partial \beta_k \partial \beta_k^\top} \right],\qquad
K_{\beta_k}:=\mathbb{E}\left[ \left(\frac{\partial  \ell_{\rm LA }}{\partial \beta_k}\right) \left(\frac{\partial  \ell_{\rm LA }}{\partial \beta_k}\right)^\top \right].
$$
\end{theorem}

Theorem~\ref{thm:an-stud} explicitly characterizes the asymptotic behavior of the estimator of regression coefficients $\beta_k$. This result plays a pivotal role in quantifying predictive uncertainty and ensuring the validity of the constructed prediction sets. 
The proofs of Theorems \ref{lem:la-consistency} and \ref{thm:an-stud} are provided in the supplementary materials.

\section{Numerical studies}
\label{sec-sim}
In this section, we conduct three simulation experiments to investigate the finite-sample performance of the proposed MEMoE across a sequence of mixture-of-experts designs. 
  We compare MEMoE   against three methods:
(1) The linear mixed model (LMM), which incorporates random effects to account for within correlations but assumes a uniform model structure across all subjects (i.e., it lacks the mixture-of-experts framework to capture group-level heterogeneity);
(2) The classical mixture-of-experts (MoE) model, which captures group heterogeneity through latent components but treats repeated measurements as statistically independent (i.e., it ignores within correlations);
(3) The random effects mixture-of-experts (REMoE) model, which incorporates both the mixture-of-experts structure and random effects, with the random effects distribution specified as zero-mean.
In the three simulation experiments,
the generated data are randomly split into two datasets: $80\%$ allocated to the training set and the remaining $20\%$ to the test set.
The entire simulation procedure is replicated $500$ times for each case. 
Additionally,  we apply the proposed method to a real-data case study.

\subsection{Simulation studies}

\textbf{Example 1:}
We generate the data from a two-component mixture-of-experts model without a random-effect term.   The latent expert label
$$
v_{i}\sim\mathrm{Muitinomial}\bigl(\pi_{1}(x_{i}),\pi_{2}(x_{i})\bigr),
\quad
\pi_{k}(x_{i})
=\frac{\exp(x_{i}^\top\alpha_{k})}
     {\sum_{\ell=1}^{2}\exp(x_{i}^\top\alpha_{\ell})},~~ k=1,2,
$$
where the gating‐network parameters
$\alpha_{1} = (6,-5,3,2,1)^\top$ and $\alpha_{2}=(-4,2,-7,5,-3)^\top$.
Conditional on $v_{i}=k$, the response variable is generated via the following model:
$$
y_{i}=x_{i}^\top\beta_{k}+\epsilon_{i}, 
$$
with regression coefficients vectors
$\beta_{1}=(3,-3,1,-1,0)^\top$ and $\beta_{2}=(-5,5,0,2,-2)^\top$. The  covariate vector
$x_{i}$ is independently generated from $\mathcal{N}(\mathbf{0},I_{5})$, and the random error $\epsilon_{i}$ is generated from  $\mathcal{N}(0,1)$.
 A total of $1500$ independent observations are sampled.
The average biases and mean squared errors (MSEs) of regression parameters and gating-network parameters for MoE, ReMoE, and MEMoE 
are presented in Table \ref{tab:simulation1_beta_alpha}. The prediction errors are exhibited in the left panel of
Figure \ref{fig:simulation-exam12}.

\begin{table}[H]
\centering
\caption{Average biases and MSEs  of parameter estimators in Example 1 and MEMoE attains the performance of oracle estimators.}
\label{tab:simulation1_beta_alpha}
\begin{tabular}{@{}lccccccc@{}}
\toprule
&\multirow{2}{*}{\centering \textbf{Parameter}} & \multicolumn{6}{c}{\textbf{Method}}\\
\cmidrule(lr){3-8}
         &             &   \multicolumn{2}{c}{ MoE }   & \multicolumn{2}{c}{   ReMoE}          &  \multicolumn{2}{c}{  MEMoE  }         \\                                  
\midrule    
         &             &  Bias  &  MSE      & Bias   &  MSE    & Bias   & MSE        \\
\cmidrule(lr){3-8}
Expert 1 & $\beta_1$   & 0.0025 & (0.0015)  & 0.0026 & (0.0015)& 0.0024 & (0.0015) \\ \addlinespace[0.1em]       
         & $\alpha_1$  & 0.0661 & (0.0085)  & 0.0793 & (0.0106)& 0.0726 & (0.0094) \\  \addlinespace[0.1em]      
Expert 2 & $\beta_2$   & 0.0049 & (0.0013)  & 0.0049 & (0.0012)& 0.0048 & (0.0013) \\  \addlinespace[0.1em]      
         & $\alpha_2$  & 0.0628 & (0.0076)  & 0.0575 & (0.0051)& 0.0542 & (0.0047) \\                            
\bottomrule
\end{tabular}
\end{table}
As expected, the standard MoE estimator exhibits negligible  bias and small variance for both expert-specific regression coefficients and gating parameters (Table~\ref{tab:simulation1_beta_alpha}). The estimates of the expert coefficients $\beta_k$ and gating parameters $\alpha_k$ obtained from the proposed MEMoE and ReMoE are unbiased and have variances comparable to those of the correctly specified MoE.
The simulation results demonstrate that, when the data-generating process is a standard MoE without random effects, introducing a mixed-effects layer into the MoE framework does not result in significant efficiency loss in either parameter estimation or predictive performance.

\textbf{Example 2:}
We generate the data from a classical linear mixed‐effects model:
\begin{equation}\label{eq:example2}
y_{ij}=x_{ij}^\top\beta + u_i + \epsilon_{ij},~~ i=1,\cdots,100,~~ j=1,\cdots, 15,
\end{equation}
where $\beta = (3,\ -3,\ 1,\ -1,\ 0)^\top$, and the covariate vector $x_{ij}$ is independently sampled from  $\mathcal{N}({\bf 0},I_5)$, in which $I_5$ denotes the five-dimensional identity matrix. 
 The random effect $u_i$ is generated from a normal distribution $\mathcal{N}(0,\tau)$, where
  the variance parameter $\tau$ takes a value of $0.01$, $0.1$, $1$, $2.5$, and $5$.
The random error  $\epsilon_{ij}$ is drawn from $\mathcal{N}(0,1)$.
The average biases and mean squared errors (MSEs) of regression parameters and gating network parameters for the MoE, ReMoE, and MEMoE 
are presented in the top panel of Table \ref{tab:examples_123}. The prediction errors are exhibited in the right panel of
Figure \ref{fig:simulation-exam12}.

When $\tau$ is less than $1$, the performance of the four methods is comparable. As $\tau$ increases, the bias and variance of MoE increase and are larger than those of the other three methods; in contrast, the biases and variances of MEMoE and ReMoE exhibit slight fluctuations and remain stable, yielding results comparable to those of the classical LMM approach. The boxplots (Figure~\ref{fig:simulation-exam12}) show prediction mean squared errors (PMSE) of the four methods. The PMSE of MoE increases with $\tau$. In contrast, the PMSEs of MEMoE, ReMoE, and LMM have no significant change as $\tau$ increases.

Example 2 serves as a key validation: when the data conform to a classical LMM framework, the MEMoE procedures closely align with the LMM benchmark in terms of both estimation and prediction accuracy. However, it is worth noting that the performance of the MoE approach declines as variability between subjects increases.
When considered alongside the findings from Example 2, these results demonstrate that MEMoE adapts well to both the presence and absence of subject-level random effects. In contrast, models that neglect either the random effects (the standard MoE) or expert heterogeneity (LMM) exhibit severe bias when deployed in misspecified modeling scenarios.

\begin{figure}[H]
    \centering
    \subfloat[Example 1\label{fig:example1}]{
        \includegraphics[width=0.45\textwidth]{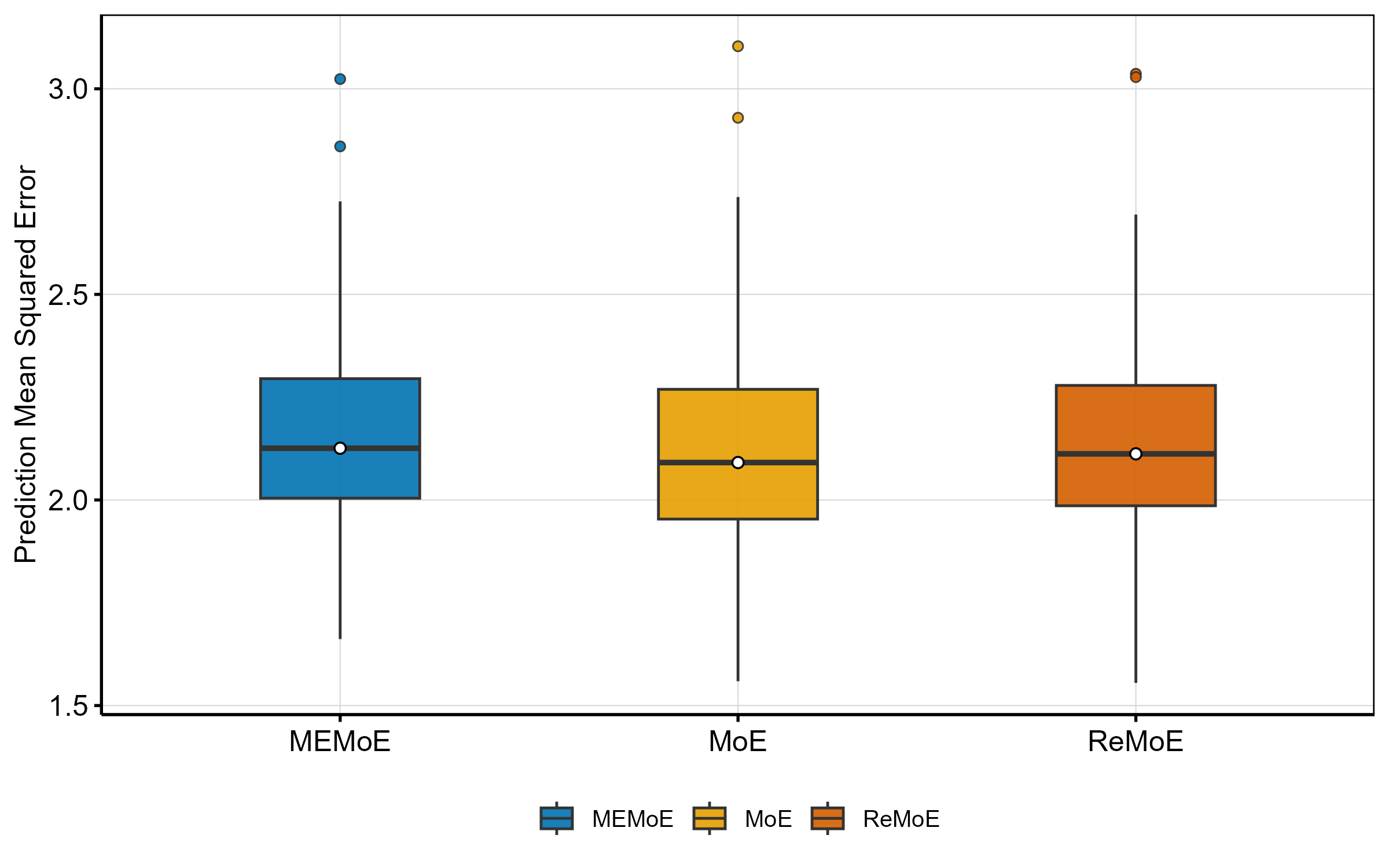}
    }
    \hfill 
    \subfloat[Example 2\label{fig:example2}]{
        \includegraphics[width=0.45\textwidth]{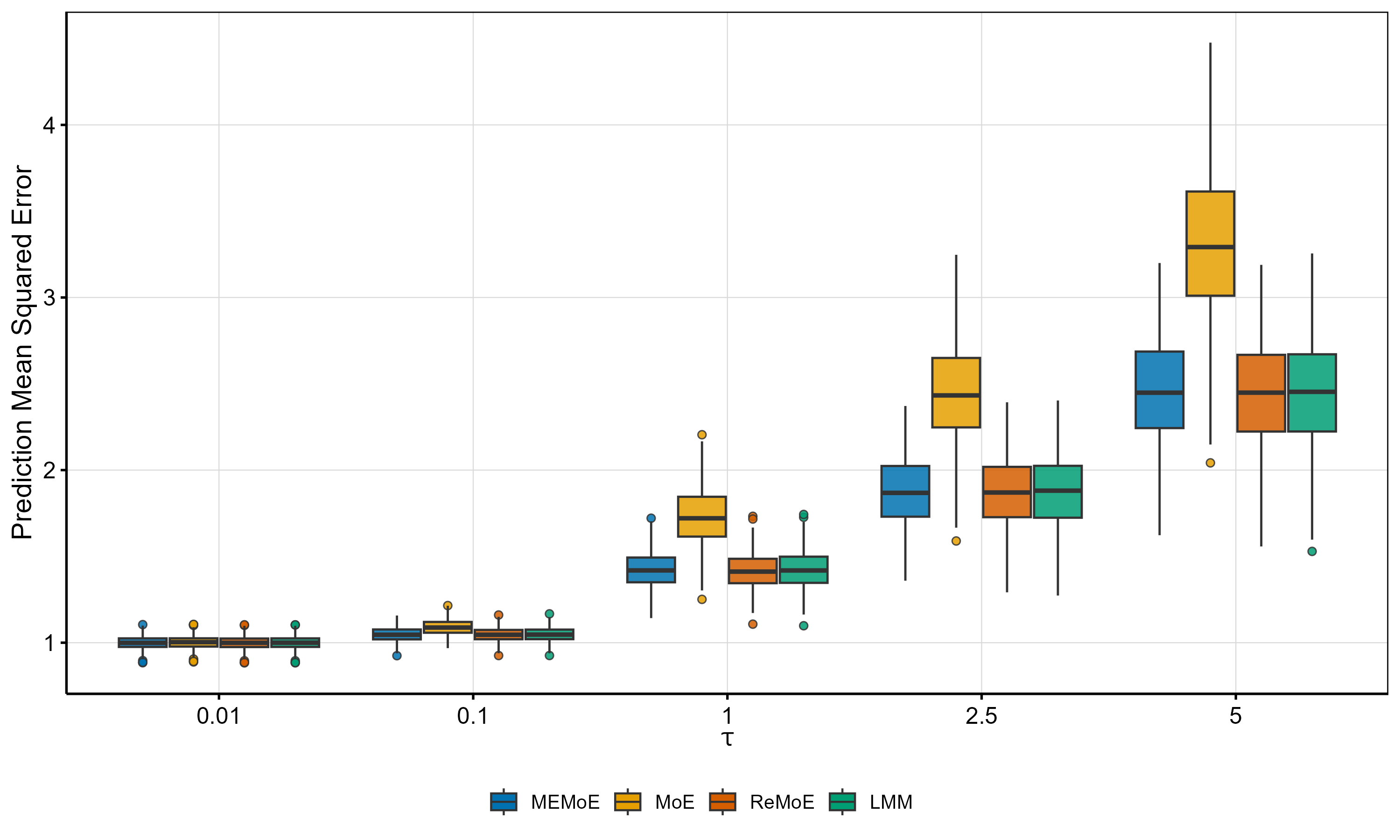}
    }
      \caption{Prediction
     mean square error under MEMoE, MoE, ReMoE and LMM 
     for Examples 1 and 2. LMM is deleted in Examples 1 because it  shows larger PMSE and all results imply that MEMoE attains oracle estimation performance.}
    \label{fig:simulation-exam12}
\end{figure}

\textbf{Example 3:}
We consider a linear mixed model mixture-of-experts framework with a Gaussian response. 
The observed data consist of $n = 100$  subjects, with each subject contributing $n_i = 15$ repeated measurements; this yields a total of $N = \sum_{i=1}^{n} n_i = 1500$ observations.
The  covariate vector $x_{ij} = (x_{ij1}, x_{ij2}, \dots, x_{ij5})^\top$ is independently drawn from a multivariate normal distribution $\mathcal{N}(\mathbf{0}, \mathbf{I}_5)$. 
Expert membership for each observation is determined by a multinomial gating mechanism, where a softmax function governs class assignment probabilities:
$$
\pi_k(x_{ij},z_{ij},\alpha) = \frac{\exp(x_{ij}^\top \alpha_k)}{\sum_{\ell=1}^K \exp(x_{ij}^\top \alpha_\ell)},\quad k=1,\cdots, K.
$$
For each observation with assigned latent expert $k$, the Gaussian response $y_{ij}$ is generated according to a linear mixed-effects model:
$$
y_{ij} = {x}_{ij}^\top \boldsymbol{\beta}_k+ u_{i} + \epsilon_{ij}.
$$

We consider the following three cases for the number of experts $K$, the random effect $u_i$, the random error $\epsilon_{ij}$, the regression coefficients  $ \boldsymbol{\beta}_k$, and the gating-network parameters
$\alpha_k$ for $k=1,\cdots, K$.

\textbf{Case 1:}
For $K=2$.
Expert-specific fixed effects and the gating-network parameters are given by:
\begin{align*}
\boldsymbol{\beta}_1 &= (3,\ -3,\ 1,\ -1,\ 0)^\top, \quad &\boldsymbol{\alpha}_1 &= (6,\ -5,\ 3,\ 2,\ 1)^\top, \\
\boldsymbol{\beta}_2 &= (-5,\ 5,\ 0,\ 2,\ -2)^\top, \quad &\boldsymbol{\alpha}_2 &= (-4,\ 2,\ -7,\ 5,\ -3)^\top.
\end{align*}
The random error $\epsilon_{ij}$ is independently sampled from $ \mathcal{N}(0, 1)$, and the random effect $u_{i}$ is generated from 
$ \mathcal{N}(0, \tau)$.

\textbf{Case 2:} The settings of the number of experts, the random error, the expert-specific fixed effects, and the gating-network parameters   are the same as those in Case 1.
The  random effect $u_{i} $ is generated from
$ \mathcal{N}(\kappa w_{i}, \tau)$, where $\omega_i = (1, \omega_{i2}, \omega_{i3}, \omega_{i4})^\top$ with $\omega_{ij} \stackrel{i.i.d.}{\sim} \mathcal{N}(0, 1)$ for $j=2,3,4$, and $\kappa = (2, 5, -1, 5)^\top$.

\textbf{Case 3:}
For $K=3$.
Expert-specific fixed effects and  the gating-network parameters are given by
\begin{align*}
\boldsymbol{\beta}_1 &= (3,\ -3,\ 1,\ -1,\ 0)^\top, \quad &\boldsymbol{\alpha}_1 &= (6,\ -5,\ 3,\ 2,\ 1)^\top, \\
\boldsymbol{\beta}_2 &= (-5,\ 5,\ 0,\ 2,\ -2)^\top, \quad &\boldsymbol{\alpha}_2 &= (-4,\ 2,\ -7,\ 5,\ -3)^\top, \\
\boldsymbol{\beta}_3 &= (1,\ -2,\ 3,\ 1,\ -4)^\top, \quad &\boldsymbol{\alpha}_3 &= (2,\ -1,\ 4,\ -3,\ 6)^\top.
\end{align*}
The  random effect $u_{i} $ is generated from $\mathcal{N}(\kappa w_{i}, \tau)$,
where $\kappa = (5, -3, 2, 1)$ and $\omega_i = (1, \omega_{i2}, \omega_{i3}, \omega_{i4})^\top$ with $\omega_{ij} \stackrel{i.i.d.}{\sim} \mathcal{N}(0, 1)$.
Conditional on the $k$th  expert, the random error  $\epsilon_{ij} \sim \mathcal{N}(0, \sigma_k^2)$. The standard deviations for three experts are $\sigma_1 = 1.0$, $\sigma_2 = 1.2$, and $\sigma_3 = 0.8$, respectively.

The average biases and mean squared errors  of the parameter estimators for fixed effects in three cases are presented in Table \ref{tab:examples_123}.
The prediction mean squared errors for three cases are displayed in panels (a), (c), and (e) of Figure \ref{fig:simulation-all}.
The empirical coverage probabilities for three cases are shown in panels (b), (d), and (f) of Figure \ref{fig:simulation-all}.

For the expert-specific regression coefficients $\beta_k$, both MEMoE and ReMoE yield essentially unbiased estimates with small variance across all values of $\tau$. In contrast, the LMM estimator, which ignores the underlying expert structure, exhibits substantial bias, reflecting its inability to recover expert-specific effects. 
The standard MoE performs well when $\tau$ is near zero but shows a marked increase in both bias and variance as $\tau$ grows, confirming the severe impact of ignoring subject-level random effects.

Across all three cases, MEMoE and ReMoE achieve the lowest and most stable prediction mean squared errors, while the performance of the standard MoE deteriorates sharply with larger $\tau$, and LMM yields intermediate but consistently inferior performance. 
Overall, Simulation 3 confirms the theoretical advantages of correctly specifying both the expert structure and the random-effects distribution, and highlights that MEMoE provides the most reliable parameter and prediction inference in heterogeneous mixed-effects settings.

Across all three simulation settings, the prediction sets of the proposed MEMoE maintain empirical $95\%$ coverage probabilities closely aligned with the nominal level across a wide range of random-effect variances. The mean interval lengths increase as $\tau$ grows.
The simulation results provide strong finite-sample evidence for the validity of the proposed prediction-set construction.

\begin{table}[H]
\centering
\caption{Biases and MSEs of  parameter estimates of $\beta$ in Examples 2 and 3 across different variances of random effects. MEMoE shows robust and best estimation performance except for the oracle estimators.}
\label{tab:examples_123}
\resizebox{\textwidth}{!}{%
\begin{tabular}{@{}lll ccccc ccccc@{}}
\toprule
 & & & \multicolumn{2}{c}{$\tau=0.01$} & \multicolumn{2}{c}{$\tau=0.1$} & \multicolumn{2}{c}{$\tau=1$} & \multicolumn{2}{c}{$\tau=2.5$} & \multicolumn{2}{c}{$\tau=5$} \\
 \cmidrule(lr){4-5} \cmidrule(lr){6-7} \cmidrule(lr){8-9} \cmidrule(lr){10-11} \cmidrule(lr){12-13} 
\textbf{Scenario} & \textbf{Parameter} & \textbf{Model} & Bias & MSE & Bias & MSE & Bias & MSE & Bias & MSE & Bias & MSE \\
\midrule

\multirow{4}{*}{\textbf{Example 2}} 
  & \multirow{4}{*}{\shortstack{ $\beta$}} 
    & LMM (oracle)       & 0.015 & 0.005 & 0.015 & 0.005 & 0.016 & 0.005 & 0.016 & 0.005 & 0.017 & 0.005 \\
  & & MoE        & 0.025 & 0.015 & 0.026 & 0.015 & 0.027 & 0.017 & 0.038 & 0.019 & 0.044 & 0.025 \\
  & & ReMoE (oracle) & 0.017 & 0.005 & 0.018 & 0.005 & 0.020 & 0.005 & 0.025 & 0.005 & 0.025 & 0.005 \\
  & & \textbf{MEMoE}     & \textbf{0.025} & \textbf{0.007} & \textbf{0.025} & \textbf{0.005} & \textbf{0.026} & \textbf{0.006} & \textbf{0.026} & \textbf{0.005} & \textbf{0.026} & \textbf{0.007} \\

\midrule
\midrule

\multicolumn{13}{l}{\textbf{Example 3:  (Cases 1--3)}} \\
\cmidrule(lr){1-13}
 & & & \multicolumn{2}{c}{$\tau=0.01$} & \multicolumn{2}{c}{$\tau=0.1$} & \multicolumn{2}{c}{$\tau=1$} & \multicolumn{2}{c}{$\tau=2.5$} & \multicolumn{2}{c}{$\tau=5$} \\
 \cmidrule(lr){4-5} \cmidrule(lr){6-7} \cmidrule(lr){8-9} \cmidrule(lr){10-11} \cmidrule(lr){12-13} 

\multirow{7}{*}{\textbf{Case 1}} 
 & \multirow{4}{*}{Expert 1 $\beta_1$} 
   & LMM           & 1.997 & 4.020 & 1.997 & 4.019 & 1.997 & 4.020 & 1.997 & 4.022 & 1.998 & 4.024 \\
 & & MoE           & 0.020 & 0.001 & 0.020 & 0.001 & 0.035 & 0.002 & 0.048 & 0.002 & 0.066 & 0.010 \\
 & & ReMoE (oracle) & 0.012 & 0.001 & 0.018 & 0.001 & 0.012 & 0.001 & 0.010 & 0.001 & 0.010 & 0.001 \\
 & & \textbf{MEMoE}& \textbf{0.015} & \textbf{0.003} & \textbf{0.014} & \textbf{0.004} & \textbf{0.019} & \textbf{0.003} & \textbf{0.020} & \textbf{0.003} & \textbf{0.021} & \textbf{0.003} \\
 \cmidrule(l){2-13}
 & \multirow{3}{*}{Expert 2 $\beta_2$} 
   & MoE           & 0.021 & 0.001 & 0.024 & 0.005 & 0.041 & 0.003 & 0.063 & 0.007 & 0.084 & 0.013 \\
 & & ReMoE & 0.021 & 0.001 & 0.019 & 0.001 & 0.016 & 0.001 & 0.017 & 0.001 & 0.017 & 0.001 \\
 & & \textbf{MEMoE}& \textbf{0.018} & \textbf{0.003} & \textbf{0.019} & \textbf{0.004} & \textbf{0.019} & \textbf{0.004} & \textbf{0.021} & \textbf{0.004} & \textbf{0.022} & \textbf{0.004} \\

\midrule

\multirow{7}{*}{\textbf{Case 2}} 
 & \multirow{4}{*}{Expert 1 $\beta_1$} 
   & LMM           & 5.955 & 2.210 & 5.938 & 2.205 & 5.954 & 2.210 & 5.945 & 2.205 & 5.958 & 2.210 \\
 & & MoE           & 0.097 & 0.403 & 0.071 & 0.411 & 0.099 & 0.421 & 0.108 & 0.406 & 0.109 & 0.415 \\
 & & ReMoE         & 0.013 & 0.044 & 0.018 & 0.047 & 0.014 & 0.047 & 0.010 & 0.046 & \textbf{0.006} & 0.048 \\
 & & \textbf{MEMoE}& \textbf{0.003} & \textbf{0.016} & \textbf{0.006} & \textbf{0.017} & \textbf{0.004} & \textbf{0.016} & \textbf{0.003} & \textbf{0.016} & 0.008 & \textbf{0.018} \\
 \cmidrule(l){2-13}
 & \multirow{3}{*}{Expert 2 $\beta_2$} 
   & MoE           & 0.092 & 0.403 & 0.096 & 0.399 & 0.118 & 0.418 & 0.104 & 0.415 & 0.102 & 0.430 \\
 & & ReMoE         & 0.011 & 0.045 & 0.019 & 0.047 & 0.013 & 0.047 & 0.011 & 0.047 & \textbf{0.003} & 0.046 \\
 & & \textbf{MEMoE}& \textbf{0.003} & \textbf{0.016} & \textbf{0.003} & \textbf{0.017} & \textbf{0.004} & \textbf{0.017} & \textbf{0.004} & \textbf{0.017} & 0.009 & \textbf{0.016} \\

\midrule

\multirow{10}{*}{\textbf{Case 3}} 
 & \multirow{4}{*}{Expert 1 $\beta_1$} 
   & LMM           & 2.641 & 8.557 & 2.641 & 8.557 & 2.641 & 8.557 & 2.642 & 8.557 & 2.642 & 8.557 \\
 & & MoE           & 1.380 & 2.483 & 1.381 & 2.489 & 1.385 & 2.510 & 1.387 & 2.533 & 1.389 & 2.561 \\
 & & ReMoE         & 0.015 & 0.014 & 0.015 & 0.014 & 0.015 & 0.014 & 0.015 & 0.014 & 0.015 & 0.014 \\
 & & \textbf{MEMoE}& \textbf{0.007} & \textbf{0.004} & \textbf{0.007} & \textbf{0.004} & \textbf{0.007} & \textbf{0.004} & \textbf{0.007} & \textbf{0.004} & \textbf{0.006} & \textbf{0.004} \\
 \cmidrule(l){2-13}
 & \multirow{3}{*}{Expert 2 $\beta_2$} 
   & MoE           & 1.679 & 3.079 & 1.681 & 3.086 & 1.685 & 3.108 & 1.689 & 3.129 & 1.693 & 3.155 \\
 & & ReMoE         & 0.014 & 0.013 & 0.017 & 0.014 & 0.016 & 0.014 & 0.016 & 0.014 & 0.016 & 0.014 \\
 & & \textbf{MEMoE}& \textbf{0.003} & \textbf{0.004} & \textbf{0.003} & \textbf{0.004} & \textbf{0.002} & \textbf{0.004} & \textbf{0.002} & \textbf{0.004} & \textbf{0.001} & \textbf{0.004} \\
 \cmidrule(l){2-13}
 & \multirow{3}{*}{Expert 3 $\beta_3$} 
   & MoE           & 1.396 & 2.643 & 1.397 & 2.645 & 1.398 & 2.654 & 1.400 & 2.666 & 1.400 & 2.681 \\
 & & ReMoE         & 0.015 & 0.012 & 0.014 & 0.012 & 0.015 & 0.012 & 0.015 & 0.012 & 0.015 & 0.012 \\
 & & \textbf{MEMoE}& \textbf{0.003} & \textbf{0.002} & \textbf{0.003} & \textbf{0.002} & \textbf{0.003} & \textbf{0.002} & \textbf{0.003} & \textbf{0.002} & \textbf{0.004} & \textbf{0.002} \\

\bottomrule
\end{tabular}%
}
\end{table}

\begin{figure}[H]
    \centering
    \subfloat[Example 3: Case 1 (PMSE)\label{fig:sim3-pmse-c1}]{
        \includegraphics[width=0.45\textwidth]{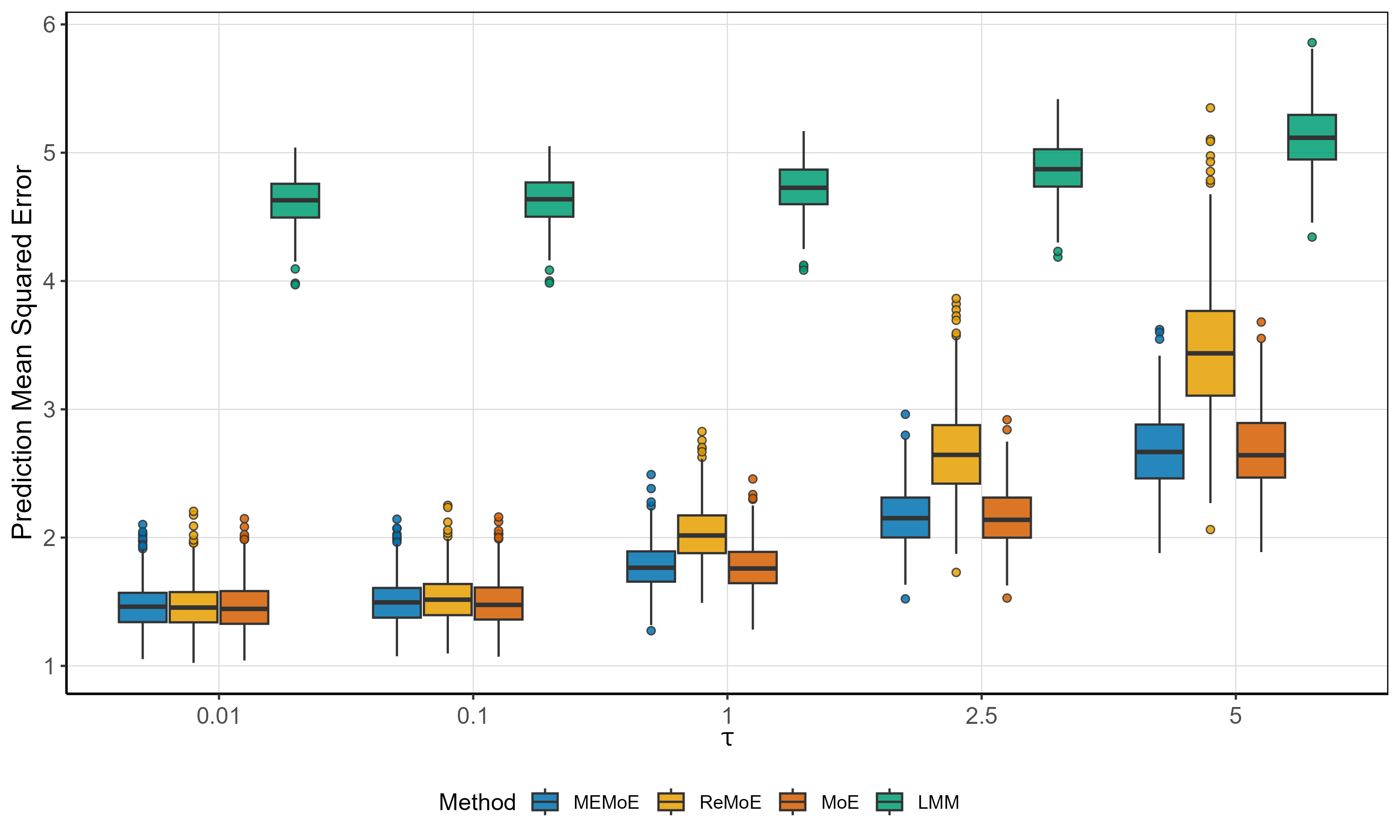}
    }
    \hfill
    \subfloat[Example 3: Case 1 (Coverage)\label{fig:sim3-cover-c1}]{
        \includegraphics[width=0.45\textwidth]{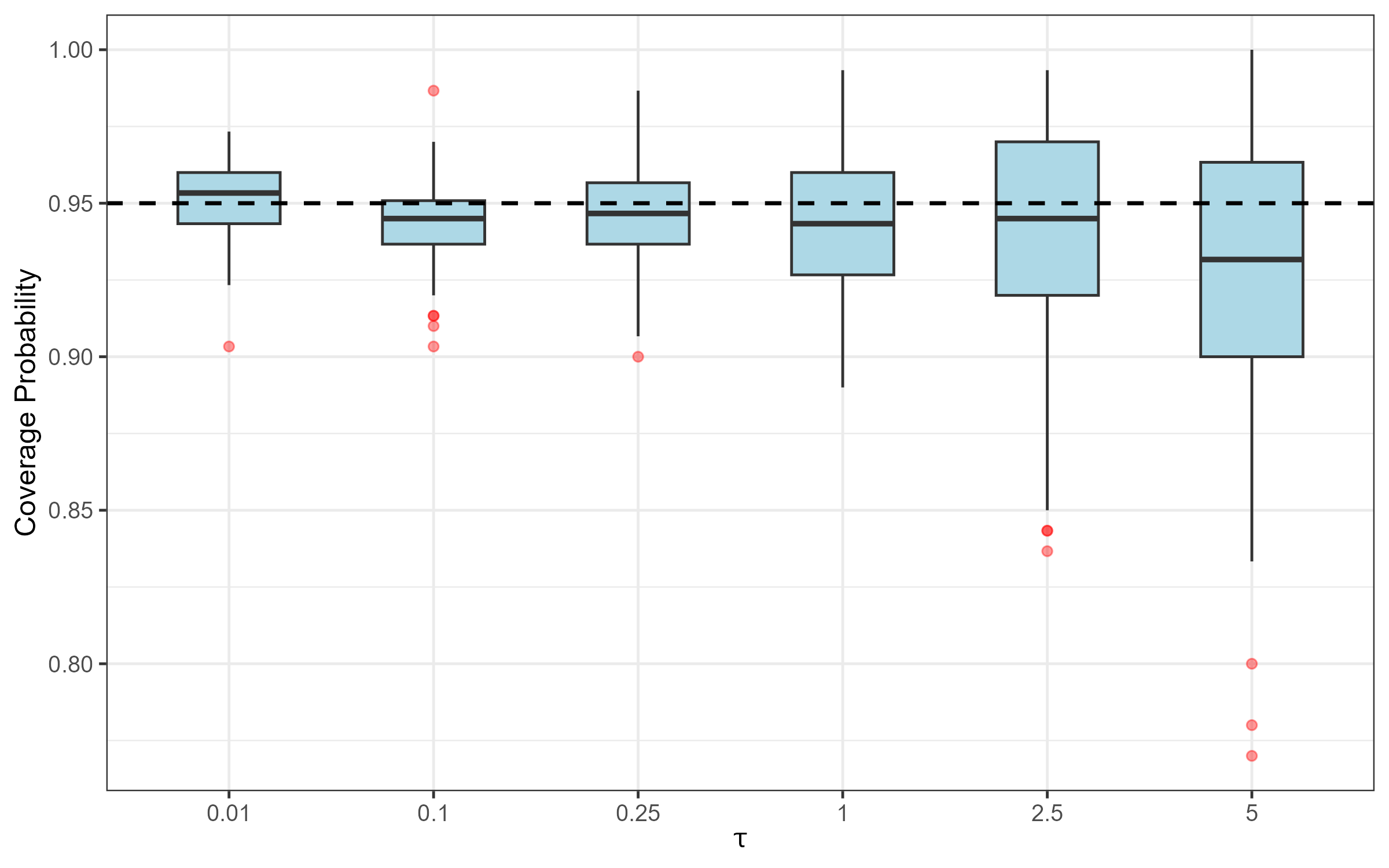}
    } 
    \vspace{-0.2cm} 
    \subfloat[Example 3: Case 2 (PMSE)\label{fig:sim3-pmse-c2}]{
        \includegraphics[width=0.45\textwidth]{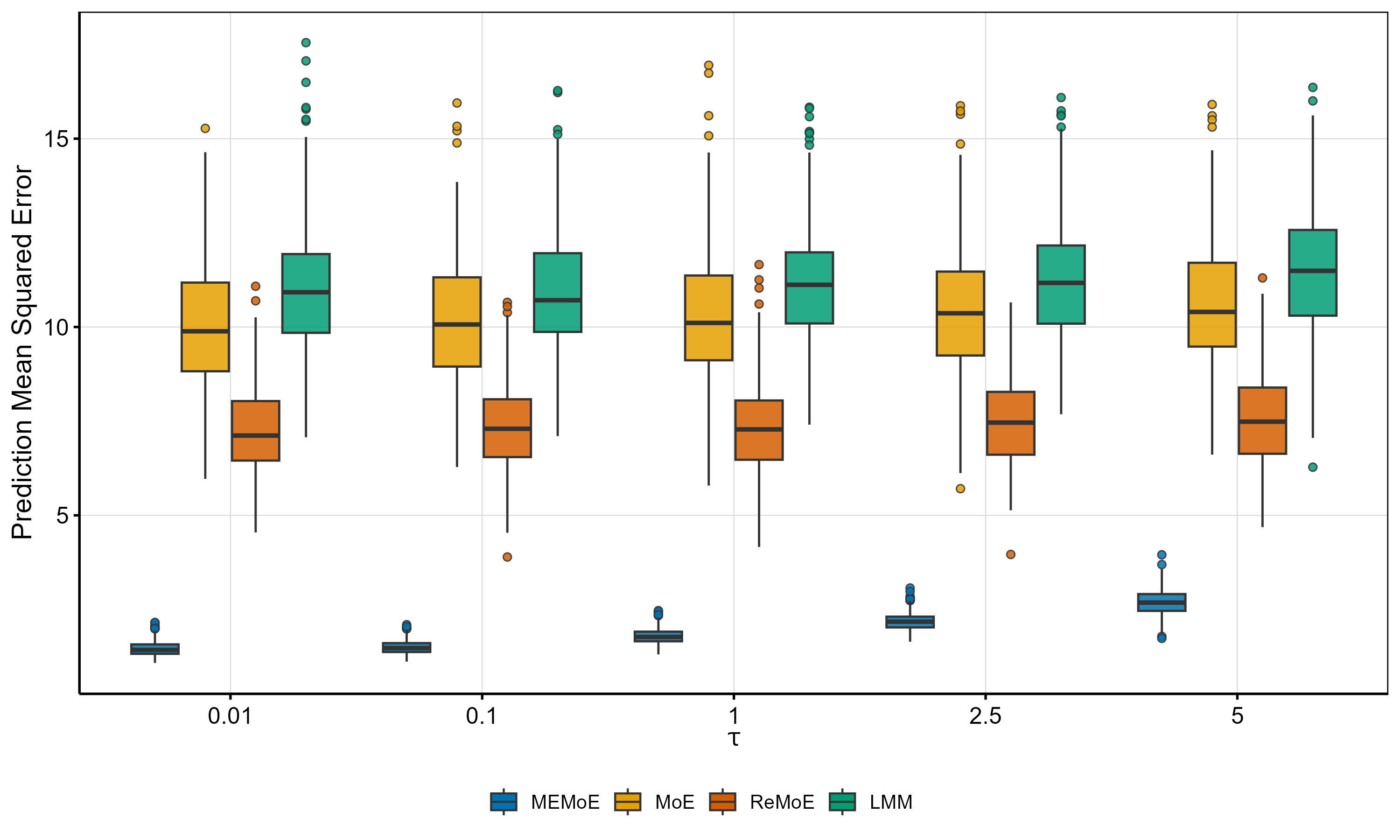}
    }
    \hfill
    \subfloat[Example 3: Case 2 (Coverage)\label{fig:sim3-cover-c2}]{
        \includegraphics[width=0.45\textwidth]{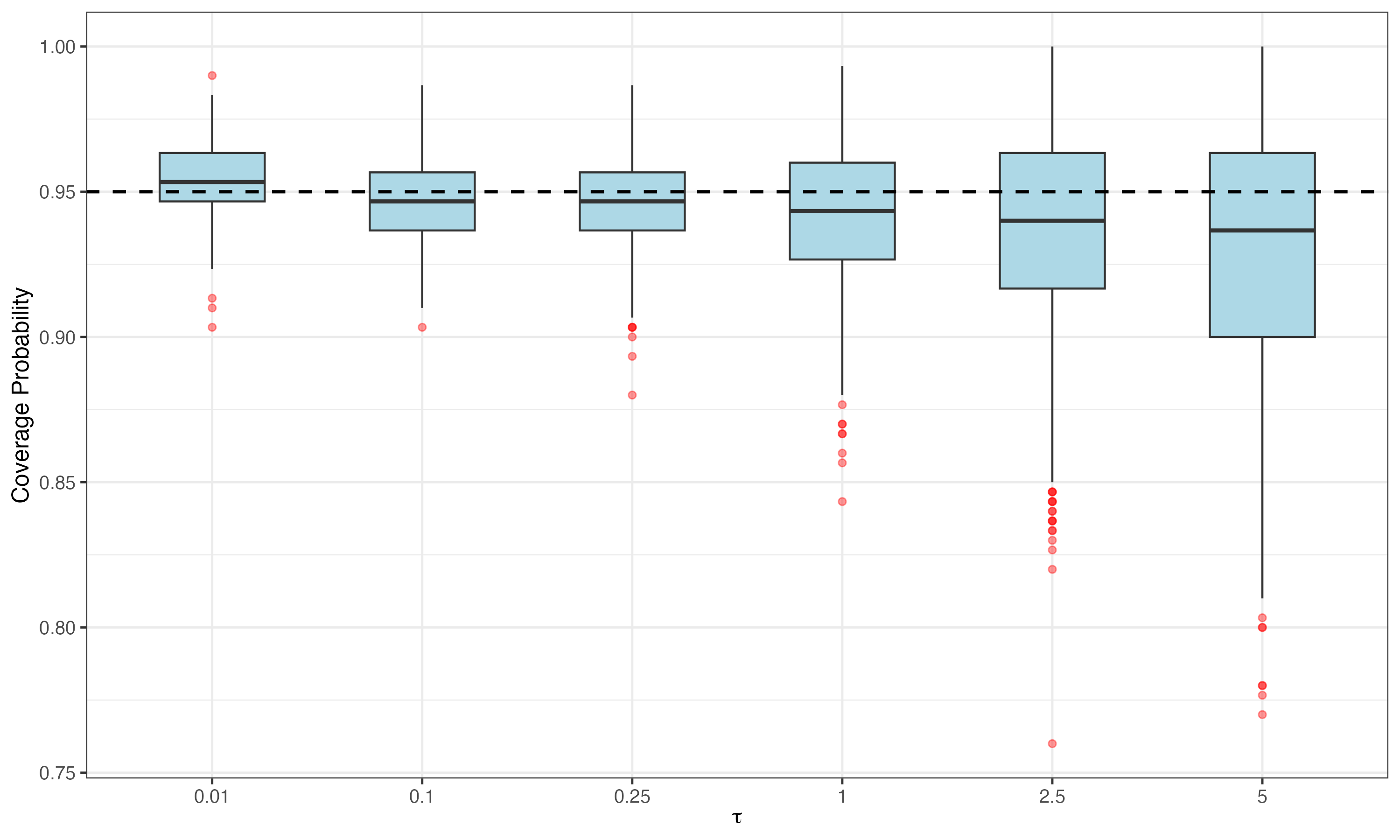}
    }
    \vspace{-0.2cm} 
    \subfloat[Example 3: Case 3 (PMSE)\label{fig:sim3-pmse-c3}]{
        \includegraphics[width=0.45\textwidth]{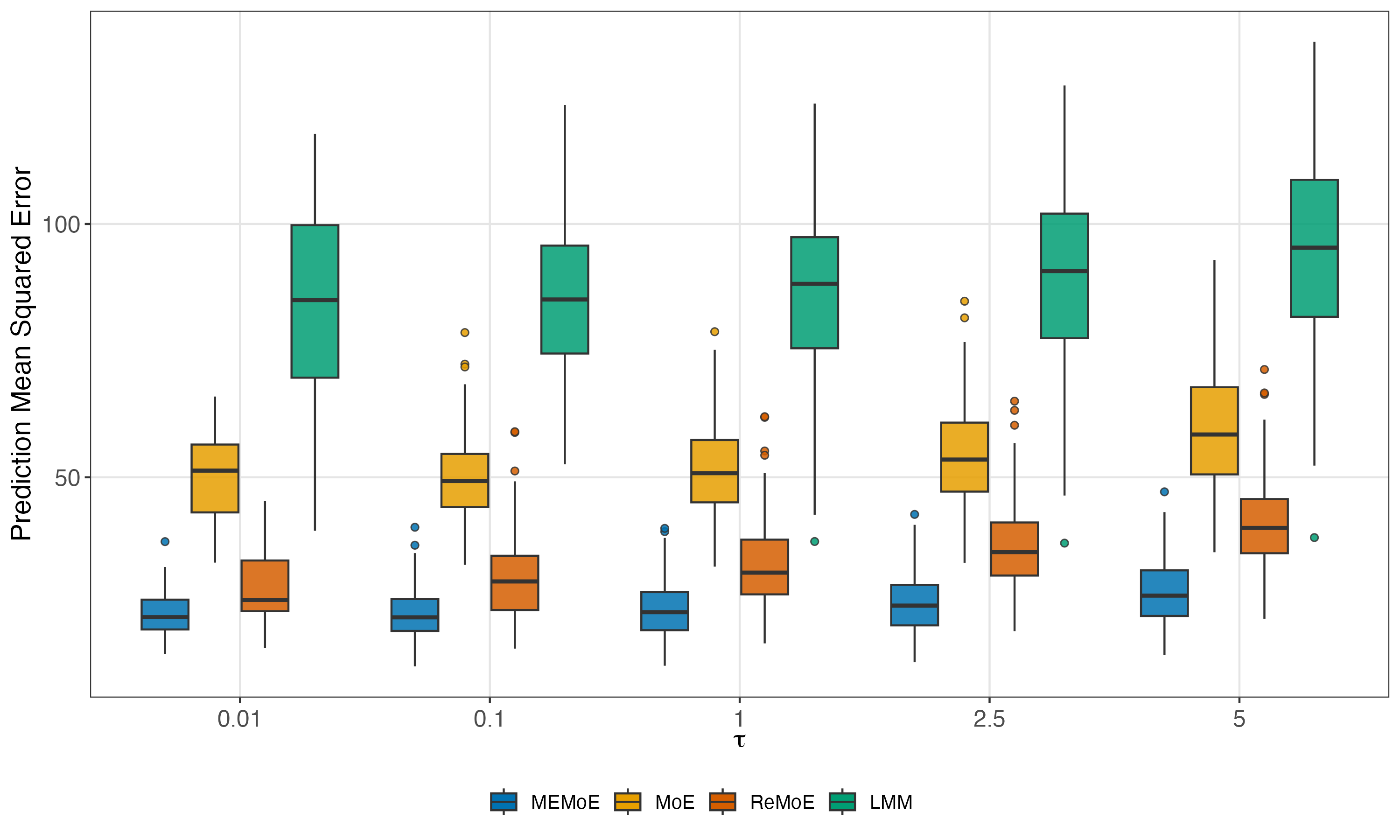}
    }
    \hfill
    \subfloat[Example 3: Case 3 (Coverage)\label{fig:sim3-cover-c3}]{
        \includegraphics[width=0.45\textwidth]{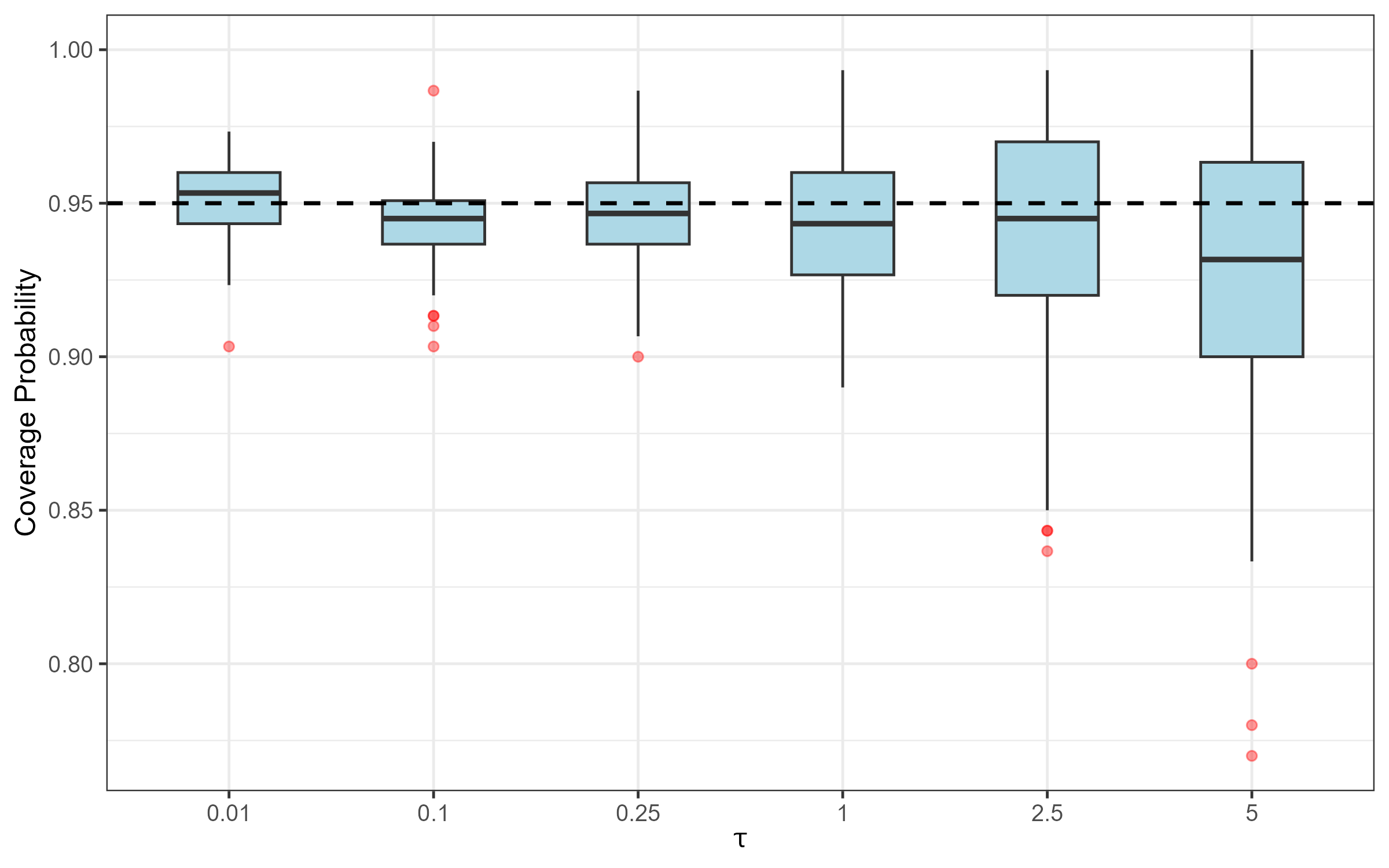}
    }
    \caption{Prediction
     mean square error (a),(c) and (e) under MEMoE, MoE, ReMoE and LMM 
     for Cases 1-3 of Example 3 and MEMoE shows the best prediction performance. MEMoE is also robust about the random effect variance in (b), (d) and (f).}
    \label{fig:simulation-all}
\end{figure}

\begin{table}[H]
\centering
\caption{Coverage probabilities (CP) and mean interval lengths of $ 95\%$ confidence intervals across $\tau$ values in Cases 1, 2, and 3.}
\label{tab:coverage_width-all-cases}
\small 
\begin{tabular}{ccccccc}
\toprule
\multirow{2}{*}{$\tau$} & \multicolumn{2}{c}{Case 1} & \multicolumn{2}{c}{Case 2} & \multicolumn{2}{c}{Case 3} \\
\cmidrule(lr){2-3} \cmidrule(lr){4-5} \cmidrule(lr){6-7}
& \multicolumn{1}{c}{CP} &  \multicolumn{1}{c}{Interval Length}
& \multicolumn{1}{c}{CP} &  \multicolumn{1}{c}{Interval Length}
& \multicolumn{1}{c}{CP} &  \multicolumn{1}{c}{Interval Length} \\
\midrule
0.01  & 0.9521 & 4.4463 & 0.9509 &  4.4592 & 0.9560 &4.7725 \\
0.10  & 0.9432 & 4.4730 & 0.9458 &  4.4891 & 0.9452 &4.7818 \\
0.25  & 0.9427 & 4.7090 & 0.9443 &  4.7204 & 0.9401 &5.0306 \\
1.00  & 0.9425 & 5.9020 & 0.9416 &  5.9086 & 0.9381 &6.2112 \\
2.50  & 0.9338 & 7.6414 & 0.9417 &  7.7088 & 0.9409 &8.0411 \\
5.00  & 0.9173 & 9.7118 & 0.9239 &  9.7735 & 0.9275 &10.0818 \\
\bottomrule
\end{tabular}
\end{table}

\subsection{Real Data Analysis }
Primary Biliary Cirrhosis (PBC) is a chronic, progressive autoimmune liver disease characterized by the destruction of intrahepatic bile ducts, leading to cholestasis, cirrhosis, and eventually liver failure. The data used in this section are from a randomized clinical trial conducted at the Mayo Clinic between January 1974 and May 1984 that compared D-penicillamine with a placebo \citep{murtaugh1994primary}. The randomized sample comprises $312$ patients with $1,945$ measurements. The number of visits per patient varies from $1$ to $16$ during a follow-up period of up to $14.1$ years. The dataset is available in the {\it survival} R package (named \text{pbcseq}).

Among the collected longitudinal biomarkers, serum bilirubin is widely recognized as the most powerful prognostic indicator of disease progression and survival.  To mitigate skewness and stabilize variance, a natural logarithm transformation is applied to serum bilirubin ($\log(\text{bili})$) as the response variable $y_{ij}$ for subsequent analyses.
To model the longitudinal disease progression while accounting for patient differences, we specify three design components: (i) observation-level covariates $x_{ij}$   for the expert mean functions; (ii) random-effects design vector $z_{ij}= (1, t_{ij})^{\top}$ for representing random intercepts and slopes and capturing within correlations, where $t_{ij}$ denotes the follow-up time in years; (iii) and subject-level covariates $w_i=(1, \text{Stage}_i)^{\top}$ for modeling the distribution of the random effects. 
The covariates maily include demographic factors: $\text{Age}_i$ and $\text{Sex}_i$ (coded as 1 for females); key biochemical markers of liver function: serum cholesterol ($\text{Chol}_{i0}$), baseline serum bilirubin ($\text{Bili}_{i0}$), albumin ($\text{Alb}_{i0}$), and prothrombin time ($\text{Protime}_{i0}$); and clinical signs of disease severity: presence of edema ($\text{Edema}_{i0}$), spiders angiomata ($\text{Spiders}_{i0}$), and ascites ($\text{Ascites}_{ij}$). Note that while $\text{Ascites}_{ij}$ is modeled as a time-varying covariate to capture disease progression, other clinical and biochemical markers are fixed at their baseline values (denoted by the subscript ``0'') to serve as baseline prognostic stratification factors. $\text{Trt}_i$ indicates the treatment group (1 for D-penicillamine, 0 for placebo).
The observation-level design vector $x_{ij}$ is constructed using a comprehensive set of feature factors. Specifically:
\begin{equation*}
x_{ij} = (1, t_{ij}, \text{Age}_{i}, \text{Chol}_{i0}, \text{Sex}_{i}, \text{Ascites}_{ij}, \text{Edema}_{i0}, \text{Trt}_{i}, \text{Bili}_{i0}, \text{Alb}_{i0}, \text{Protime}_{i0}, \text{Spiders}_{i0})^{\top}.
\end{equation*}
 Excluding time variable $t_{ij}$,  all predictors are fixed at enrollment. This approach ensures temporal precedence and reduces endogeneity bias caused by biomarkers that change with disease activity \citep{hernan2004structural}. Specifically, all continuous covariates are standardized.

The histogram (Figure \ref{fig:distribution:LMM}) of the residuals from the LMM exhibits a clear bimodal pattern, with two local peaks around the zero line and a noticeable dip in the middle. This deviation from unimodality indicates the presence of unobserved heterogeneity within the patient population. 
Therefore, given the implication of potential latent subgroups, we adopt the MEMoE framework to model this hidden heterogeneity.
\begin{figure}[ht]  
    \centering
    \includegraphics[width=0.8\textwidth]{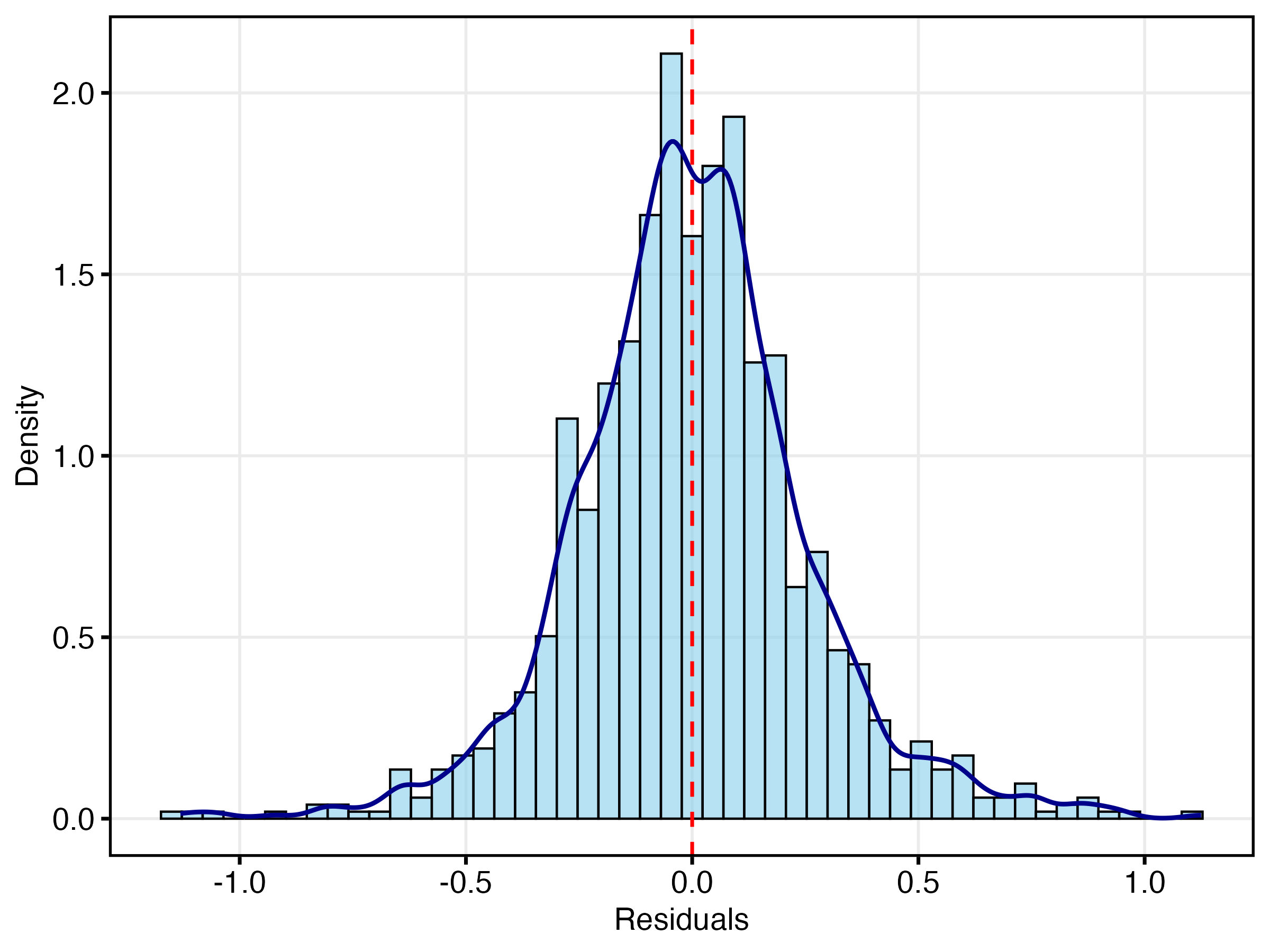} 
    \caption{The histogram  of  residuals from the LMM.}  
    \label{fig:distribution:LMM}
    \vspace{-0.2cm} 
\end{figure}

The data are randomly partitioned into $80\%$ for training and $20\%$ for testing. 
The five-fold cross-validation is conducted on the training set to select the optimal number of experts, with the training set split into four folds for model training and one for validation in each iteration. For each candidate number of experts 
$K$ ranging from 1 to 5, model parameters are estimated using the four training folds, and the corresponding predictive error is computed on the validation fold. 
Let $\hat{\pi}_k({x}_{ij}, z_{ij})$ and $\hat{{\beta}}_k$ denote the estimates of the gating probabilities and the expert-specific regression coefficients, respectively. 
We use
$$
\hat{y} = ({x}_{\mathrm{new}})^\top \hat{{\beta}}_{\hat{k}} + (\boldsymbol{z}_{\mathrm{new}})^\top \hat{{\kappa}}{w}_{\mathrm{new}} 
$$
to make predictions at a given ${x}_{\mathrm{new}}$, ${z}_{\mathrm{new}}$, ${w}_{\mathrm{new}}$,
where $\hat{k} = \operatorname*{arg\,max}_{k \in \{1,\dots,K\}} \; 
\hat{\pi}_k({x}_{\mathrm{new}},{z}_{\mathrm{new}})$ is the class with highest estimated probability. 
The cross-validated root mean squared error indicates that $K=3$ yields the lowest predictive error.
Subsequently, the full training set is used to re-estimate the model parameters using the selected optimal $K$. Table~\ref{tab:memoe-pbc-params} presents the expert-specific estimates of parameters in the three-expert MEMoE model.

\begin{table}[H]
\centering
\caption{Parameter estimates for the MEMoE model ($K=3$) on PBC data, and $^{*}$ indicates  $p$-value is less than $0.05$.}
\label{tab:memoe-pbc-params}
\small
\begin{tabular}{lrrr}
\toprule
\textbf{Parameter} & \textbf{Expert 1} & \textbf{Expert 2} & \textbf{Expert 3} \\
\midrule
\multicolumn{4}{l}{\textit{Fixed Effects ($\beta_k$)}} \\
Intercept                   & 0.393$^{*}$     & $-$2.375$^{*}$   & 1.308$^{*}$ \\
Time ($t_{ij}$)             & $-$0.030$^{*}$  & $-$0.098$^{*}$   & $-$0.041$^{*}$ \\
Baseline Age                & $-$0.048$^{*}$  & $-$0.486$^{*}$   & $-$0.096$^{*}$ \\
Baseline Albumin            & $-$0.093$^{*}$  & $-$0.573$^{*}$   & $-$0.242$^{*}$ \\
Baseline Bilirubin          & 1.496$^{*}$     & 0.673$^{*}$      & 0.698$^{*}$ \\
Baseline Cholesterol        & 0.100$^{*}$     & 0.411$^{*}$      & 0.042$^{*}$ \\
Baseline Prothrombin Time   & 0.021$^{*}$     & 0.467$^{*}$      & $-$0.030$^{*}$ \\
Sex (Female)                & $-$0.276$^{*}$  & 0.228$^{*}$      & $-$0.593$^{*}$ \\
Ascites                     & $-$0.149$^{*}$  & 1.556$^{*}$      & 0.010 \\
Edema score                 & 0.066$^{*}$     & $-$2.497$^{*}$   & $-$0.018 \\
Spiders                     & 0.123$^{*}$     & $-$0.522$^{*}$   & 0.212$^{*}$ \\
Treatment                   & 0.022           & $-$0.197$^{*}$   & 0.141$^{*}$ \\
\midrule
$\sigma_k^2$                & 0.090           & 0.041            & 0.058 \\
\bottomrule
\end{tabular}
\end{table}

The results indicate the existence of three distinct latent progression patterns, each characterized by specific factor drivers. Expert 1 represents a ``Biochemical Instability'' phenotype, exhibiting the strongest positive association with Baseline Bilirubin (1.496) among the experts. This group also demonstrates the largest variance ($\sigma_k^2 = 0.090$), reflecting a high degree of clinical unpredictability driven by biochemical derangement rather than physical signs.
Expert 2 characterizes a ``Clinical Decompensation'' phenotype. This component is dominated by severe clinical manifestations, showing the most significant coefficients for Ascites ($1.556$) and Edema ($-2.497$), as well as the strongest negative association with Baseline Albumin ($-0.573$). The substantial impact of these markers suggests this subgroup represents patients with advanced structural damage and synthetic failure. Notably, this is the only subgroup where Treatment shows a significant negative association ($-0.197$), suggesting a potential differential response.
Expert 3 captures a ``Demographic-Driven'' phenotype. Unlike the other groups, this pattern shows strong sensitivity to Sex ($-0.593$) and a high baseline intercept. In contrast to Expert 2, clinical signs of fluid retention (Ascites, Edema) contribute minimally (0.010 and $-$0.018, respectively) and are not statistically significant. This suggests a subgroup in which disease trajectory is influenced more by demographic factors and baseline metabolic state (Bilirubin 0.698) than by overt clinical decompensation.

The proposed MEMoE achieved the lowest predictive error (RMSE = $0.756$), outperforming both the standard MoE (RMSE = $0.835$) and the single-component LMM (RMSE = $0.990$). This sequential reduction in RMSE shows that accounting for within correlations via random effects—beyond merely capturing population heterogeneity—is essential for accurate PBC progression modeling.

Furthermore, we construct the $95\%$ predictive sets $(\Omega_{0.05}(x_{\text{new}}))$ and record whether $y_{\text{new}}\in \Omega_{0.05}(x_{\text{new}})$  for $100$ randomly selected test observations (see Figure \ref{fig:placeholder}).
The overall cover probability is 94.76\%.
Most prediction sets are contiguous, while the MEMoE adaptively produces disjoint intervals (visible as gaps in the red lines) in regions where the predictive distribution is highly skewed or exhibits multiple modes. 
This flexibility reflects genuine ambiguity about future $\log(\text{bilirubin})$ levels, yet even in these complex cases, the vast majority of observations are correctly covered by the prediction sets.

\begin{figure}[H]
    \centering
    \includegraphics[width=0.85\linewidth]{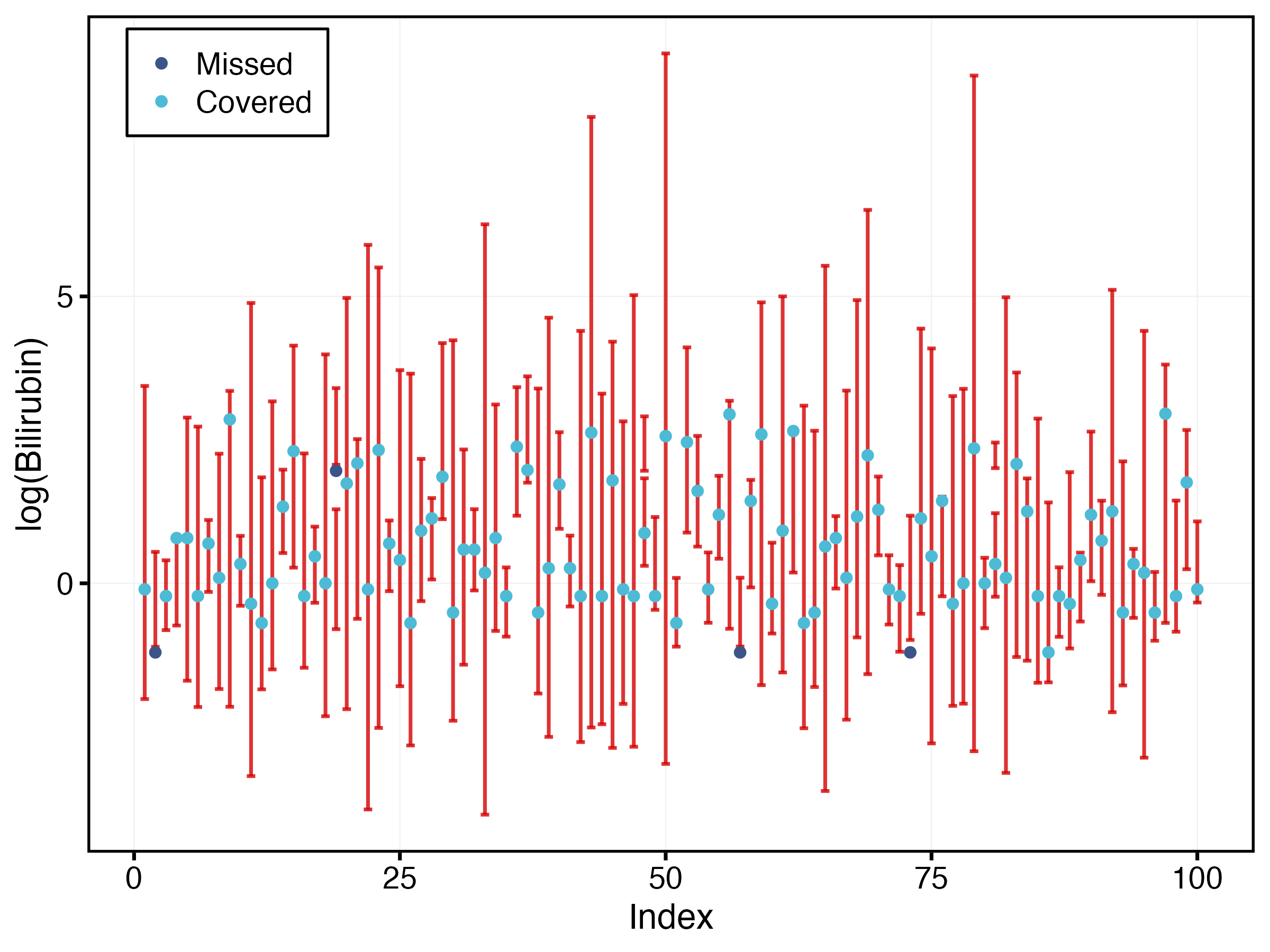}
    \vspace{-0.5cm}
    \caption{Plot of MEMoE 95\% prediction sets for the PBCseq bilirubin. The red bars are  prediction sets for each observation. The dark and light blue points correspond to the covered and missed  $\log(\text{bilirubin})$ values by the prediction sets, respectively. }
    \label{fig:placeholder}
    
\end{figure}

\section{Conclusions}\label{sec-conc}
In this paper, we propose a MEMoE model that combines subject-level random effects with a flexible expert structure and a  gating function. The  MEMoE framework provides a unified approach that simultaneously addresses between-subject heterogeneity and within correlations by integrating a covariate-dependent gating function with expert-specific linear mixed-effects models. To overcome the computational intractability of the likelihood function, we develop a robust estimation scheme that combines a Laplace approximation for latent-variable integration with a generalized EM algorithm augmented by a majorize-minimize strategy. This method ensures numerical stability and the accuracy of the model parameter estimates. Simulation studies and the real data application show that MEMoE can effectively identify hidden dynamic regimes and significantly improve predictive accuracy for new samples compared to LMMs and standard MoE. 
A key innovation in our framework is the construction of prediction sets. Unlike traditional interval estimates, the proposed method can easily handle multimodal predictive distributions, which commonly arise in heterogeneous populations, thereby facilitating more comprehensive uncertainty quantification. Despite these positive aspects, the current model has some limitations. The Laplace approximation can introduce bias when the number of longitudinal observations per subject is small, and the non-convexity of the objective function poses challenges for avoiding local optima. 
In summary, MEMoE bridges the gap between the interpretability of mixed-effects models and the flexibility of finite mixture models, providing a valuable and principled tool for complex longitudinal data analysis.


\section{Disclosure statement}\label{disclosure-statement}

The authors report there are no competing interests to declare.

\section{Data Availability Statement}\label{data-availability-statement}

The data used in this study are available in the \texttt{[survival]} R package.

\phantomsection\label{supplementary-material}
\bigskip

\begin{center}

{\large\bf SUPPLEMENTARY MATERIAL}

\end{center}



Detailed mathematical proofs for Theorems 1 and 2 are provided in the supplementary material. 


\bibliography{Bibliography-MM-MC.bib}

\newpage
\section*{Proof Theorem 1}
We define the population objective function as the expected Laplace-approximated log-likelihood for a representative subject:$$M(\Psi):= \mathbb{E}\left[ \frac{1}{N} \sum_{i=1}^N \ell_{\text{LA},i}(\Psi) \right],$$
where
$$\ell_{\text{LA},i}(\Psi) = h_i(\hat{u}_i(\Psi); \Psi) + \frac{q}{2}\log(2\pi) - \frac{1}{2}\log|H_i(\Psi)|.$$
Correspondingly, the sample objective function is:
$$M_N(\Psi) := \frac{1}{N} \sum_{i=1}^N \ell_{\text{LA},i}(\Psi).$$



The norm $\|\cdot\|$ in these supplementary materials is defined as the $L_2$ norm. For a vector $v \in \mathbb{R}^d$,
$\|v\| = \sqrt{v^\top v}=(\sum_{k=1}^d v_k^2)^{1/2}$;  for a matrix $A \in \mathbb{R}^{m \times d}$, $\|A\| := \sup_{\|v\|=1} \|Av\| = \sqrt{\lambda_{\max}(A^\top A)}$, where $\lambda_{\max}$ is the largest eigenvalue.
To prove Theorem 1, we give the following Propositions 1 and 2.

\paragraph{ Proposition 1}
Under conditions (A1)-- (A4), the class of functions $\{\ell_{\text{LA},i}(\Psi):\Psi\in\Theta\}$, as $N \to \infty$:
$$\sup_{\Psi \in \Theta} \left|M_N(\Psi) - M(\Psi)\right| \xrightarrow{p} 0.$$

\noindent{\it Remark:}
By the Glivenko-Cantelli theorem, it suffices to verify that the function class $\mathcal{F} = \{\ell_{\text{LA},i}(\Psi): \Psi \in \Theta\}$ satisfies: each function $\Psi \mapsto \ell_{\text{LA},i}(\Psi)$ is continuous on the compact set $\Theta$; there exists a random variable $G_i$ with $\sup_{\Psi \in \Theta}|\ell_{\text{LA},i}(\Psi)| \le G_i$ and $\mathbb{E}[G_i] < \infty$.




\begin{proof}

We first construct an integrable envelope function that does not depend on $\Psi$. By the eigenvalue bounds in Condition (A3) and the compactness of the parameter space $\Theta$ in Condition (A2), all parameter components $(\alpha,\beta,\sigma^2,\kappa,\Sigma)$ are uniformly bounded over $\Psi\in\Theta$. The eigenvalues of $\Sigma$ are uniformly bounded away from zero and infinity. 

Condition (A2) imposes uniform bounds on all model parameters, and Condition (A1) ensures finite second moments for the covariates. 
Therefore, the unconditional second moment of the response
$$
\mathbb{E}[y_{ij}^2] \;\le\; C' \left( 1 + \mathbb{E}\|x_{ij}\|^2 + \mathbb{E}\|z_{ij}\|^2 + \mathbb{E}\|w_i\|^2 \right) \;<\; \infty.
$$


We will prove that $\ell_{\mathrm{LA},i}(\Psi)$ is continuous on the parameter space $\Theta$.
Recall that
$$\ell_{\text{LA},i}(\Psi) = h_i(\hat{u}_i(\Psi); \Psi) + 
\frac{q}{2}\log(2\pi) - \frac{1}{2}\log|H_i(\Psi)|,$$
where $\hat{u}_i(\Psi) = \arg\max_u h_i(u; \Psi)$ and $H_i(\Psi) = -\nabla_u^2 h_i(\hat{u}_i(\Psi); \Psi)$.

We first establish the continuity of the mapping $\Psi \mapsto \hat{u}_i(\Psi)$. Specifically, $\hat{u}_i(\Psi)$ denotes the unique solution to the first-order condition $\nabla_u h_i(u, \Psi) = 0$. Under the smoothness conditions imposed by the model, the function $F(u, \Psi):= \nabla_u h_i(u, \Psi)$ is continuously differentiable.
For an arbitrary  fixed $\Psi_0 \in \Theta$, let $u_0 = \hat{u}_i(\Psi_0)$. The partial derivative of $F(u, \Psi)$ with respect to $u$, evaluated at $(u_0, \Psi_0)$, is the Hessian matrix of the joint log-likelihood function $h_i(\cdot,\cdot)$:
$$
\frac{\partial}{\partial u} F(u, \Psi)|_{(u_0, \Psi_0) }= \nabla_u^2 h_i(u_0, \Psi_0) = -H_i(\Psi_0).
$$
$H_i(\Psi_0)$ is positive definite, and hence $\nabla_u^2 h_i(u_0; \Psi_0) = -H_i(\Psi_0)$ is negative definite and thus  invertible.
By the Implicit Function Theorem, there exists a neighborhood $U$ of $\Psi_0$ and a unique continuously differentiable function $g: U \to \mathbb{R}^q$ such that $F(g(\Psi), \Psi) = 0$ for all $\Psi \in U$. Given the global uniqueness of the maximizer $\hat{u}_i(\Psi)$, we identify $\hat{u}_i(\Psi) = g(\Psi)$ on $U$. Since $\Psi_0$ is chosen arbitrarily, the mapping $\Psi \mapsto \hat{u}_i(\Psi)$ is continuous over the entire compact set $\Theta$. 

The function $h_i(u; \Psi)$ is differentiable with respect to $u$ and $\Psi$, and the mapping $(u, \Psi) \mapsto h_i(u; \Psi)$ is jointly continuous. Since the mapping $\Psi \mapsto \hat{u}_i(\Psi)$ is continuous, 
$\Psi \mapsto h_i(\hat{u}_i(\Psi); \Psi)$
is continuous on $\Theta$.
The Hessian $(u, \Psi) \mapsto \nabla_u^2 h_i(u; \Psi)$ is continuous (as the second derivative of a three-times differentiable function). Since $H_i(\Psi) = -\nabla_u^2 h_i(\hat{u}_i(\Psi); \Psi)$, similarly, the mapping $\Psi \mapsto H_i(\Psi)$ is continuous on $\Theta$.

Since $H_i(\Psi)$ is positive definite,  all its eigenvalues are strictly positive. Since $\Psi \mapsto H_i(\Psi)$ is continuous,  and $\Theta$ is compact, the eigenvalue function $\Psi \mapsto \lambda_{\min}(H_i(\Psi))$ is continuous on the compact set. By the extreme value theorem,
$$
c:= \inf_{\Psi \in \Theta} \lambda_{\min}(H_i(\Psi)) 
= \min_{\Psi \in \Theta} \lambda_{\min}(H_i(\Psi)) > 0.
$$
Therefore, for all $\Psi \in \Theta$:
$$
|H_i(\Psi)| = \prod_{j=1}^q \lambda_j(H_i(\Psi)) 
\ge \lambda_{\min}(H_i(\Psi))^q \ge c^q > 0.
$$
Since $|H_i(\Psi)|$ is strictly positive and continuous on $\Theta$, and the logarithm function is continuous on $(0, \infty)$, the mapping
$$\Psi \mapsto \log|H_i(\Psi)|$$
is continuous on $\Theta$.

Finally, the function $\ell_{\text{LA},i}(\Psi)$ is a linear combination of continuous functions:
$$
\ell_{\text{LA},i}(\Psi) = \underbrace{h_i(\hat{u}_i(\Psi); \Psi)}_{\text{continuous}} 
+ \underbrace{\frac{q}{2}\log(2\pi)}_{\text{constant}} 
- \frac{1}{2}\underbrace{\log|H_i(\Psi)|}_{\text{continuous}}.
$$
Therefore, $\ell_{\text{LA},i}(\Psi)$ is continuous on $\Theta$.

We now prove that there exists a random variable $G_i$ such that 
$$
\sup_{\Psi \in \Theta} |\ell_{\mathrm{LA},i}(\Psi)| \le G_i, \quad \text{and} \quad \mathbb{E}[G_i] < \infty.
$$
The maximizer $\hat u_i(\Psi) = \arg\max_{u_i}h_i(u_i; \Psi) $ satisfies the first-order condition 
$$ 
\nabla_{u_i} h_i( u_i ; \Psi) = -\Sigma^{-1} \left( u_i - \kappa w_i \right) + \sum_{j=1}^{n_i} \sum_{k=1}^{K} \gamma_{ijk}(\Psi) \frac{z_{ij} \left ( y_{ij} - x_{ij}^\top \beta_k - z_{ij}^\top u_i \right)}{\sigma_k^2} = 0,
$$
where $\gamma_{ijk}(\Psi) \in (0,1)$ and $ \sum_{k=1}^K \gamma_{ijk}=1$. Rearrange the first-order condition, and we obtain
\begin{equation}\label{th1:eq1}
    \Sigma^{-1} \hat{u}_i+\sum_{j, k} \gamma_{ij k}(\Psi) \frac{z_{ij} z_{ij}^{\top}}{\sigma_{k}^{2}} \hat{u}_i=\Sigma^{-1} \kappa w_i+\sum_{j, k} \gamma_{j k}(\Psi) \frac{z_{ij}\left(y_{ij}-x_{ij}^{\top} \beta_{k}\right)}{\sigma_{k}^{2}}.
\end{equation}
According to
$$
H_i(\Psi) =  \Sigma^{-1} + \sum_{j=1}^{n_i} \sum_{k=1}^K \frac{\pi_k(x_{ij}, z_{ij}; \alpha) \varphi_{ijk}(\hat{u}_i)}{\sum_{\ell=1}^K \pi_\ell(x_{ij}, z_{ij}; \alpha) \varphi_{ij\ell}(\hat{u}_i)} \frac{1}{\sigma_k^2} z_{ij}z_{ij}^\top,
$$
the equation (\ref{th1:eq1}) becomes 
$$
H_i(\Psi)\hat{u}_i=\Sigma^{-1} \kappa w_i+\sum_{j, k} \gamma_{j k}(\Psi) \frac{z_{ij}\left(y_{ij}-x_{ij}^{\top} \beta_{k}\right)}{\sigma_{k}^{2}}.
$$
Under Condition (A3), 
the eigenvalues of \(\Sigma\) are bounded as $c_1\le\lambda_{\min}(\Sigma)\le\lambda_{\max}(\Sigma)\le c_2$. According to Condition (A2), the variances satisfy
$0<\sigma_{\min}^2\le\sigma_k^2\le\sigma_{\max}^2<\infty$ for all $k$. It follows that
$$
H_i(\Psi)\succeq\Sigma^{-1}
\quad\Longrightarrow\quad
\lambda_{\min}(H_i(\Psi))\ge\lambda_{\min}(\Sigma^{-1})
=\frac{1}{\lambda_{\max}(\Sigma)}\ge \frac{1}{c_2},
$$
and hence
$$
\|H_i(\Psi)^{-1}\| =\frac{1}{\lambda_{\min}(H_i(\Psi))} \le c_2
$$
uniformly over \(\Psi\in\Theta\).


Taking norms on both sides of the normal equation and using the triangle inequality, we obtain
$$
\|\hat{u}_i(\Psi) \| \leq\left\|H_i(\Psi)^{-1}\right\| \left(\|\Sigma^{-1} \kappa w_i\|+ \left\|\sum_{j, k} \gamma_{j k}(\Psi) \frac{z_{ij}\left(y_{ij}-x_{ij}^{\top} \beta_{k}\right) }{\sigma_{k}^{2}}\right\|\right),
$$
Under Conditions  (A2) and (A3), there exists a constant
\(C_\kappa>0\) such that \(\|\kappa\|\le C_\kappa\), and similarly
\(\|\beta_k\|\le B\) for all \(k\), where $B:=\sup _{k}\left\|\beta_{k}\right\|<\infty$. Moreover,
\(\|\Sigma^{-1}\|\le c_1^{-1}\) and \(\sigma_k^2\ge\sigma_{\min}^2\). Hence
$$
\left\|\Sigma^{-1} \kappa w_i\right\| \leq\left\|\Sigma^{-1}\right\|\|\kappa\|\|w_i\| \leq c_{1}^{-1} C_{\kappa}\|w_i\|,
$$
we can obtain 
$$
\left\|\sum_{j, k} \gamma_{j k}(\Psi) \frac{z_{ij}\left(y_{ij}-x_{ij}^{\top} \beta_{k}\right)}{\sigma_{k}^{2}}\right\| \leq \frac{1}{\sigma_{\min }^{2}} \sum_{j=1}^{n_i}\left\|z_{ij}\right\|\left(\left|y_{ij}\right|+\max _{k}\left|x_{ij}^{\top} \beta_{k}\right|\right) \le \frac{1}{\sigma_{\min}^2}
     \sum_{j=1}^{n_i}\|z_{ij}\|\bigl(|y_{ij}|+B\|x_{ij}\|\bigr),
$$
where $\sigma_{\text {min }}^{2}>0$ is the lower bound of the variances, and $\left|x_{ij}^{\top} \beta_{k}\right| \leq\left\|x_{ij}\right\|\left\|\beta_{k}\right\| \leq B\left\|x_{ij}\right\|$.
Altogether,
$$
\|\hat u_i(\Psi)\| \leq c_{2}\left(c_{1}^{-1} C_{\kappa}\|w_i\|+\frac{1}{\sigma_{\min }^{2}} \sum_{j=1}^{n_i}\left\|z_{ij}\right\|\left(\left|y_{ij}\right|+B\left\|x_{ij}\right\|\right)\right),
$$
Using $(a + b)^2 \le 2(a^2 + b^2)$:
$$
\|\hat{u}_i(\Psi)\|^2 
\le 2c_2^2\left[\frac{C_\kappa^2}{c_1^2}\|w_i\|^2 + 
\frac{1}{\sigma_{\min}^4}\left(\sum_{j=1}^{n_i}\|z_{ij}\|(|y_{ij}| + B\|x_{ij}\|)\right)^2\right].
$$
For the second term, using the Cauchy-Schwarz inequality:
\begin{align*}
\left(\sum_{j=1}^{n_i}\|z_{ij}\|(|y_{ij}| + B\|x_{ij}\|)\right)^2 
&\le \sum_{j=1}^{n_i}\|z_{ij}\|^2\cdot \sum_{j=1}^{n_i}\left(|y_{ij}| + B\|x_{ij}\|\right)^2 \notag\\
&\le \sum_{j=1}^{n_i}\|z_{ij}\|^2 \cdot \sum_{j=1}^{n_i} 2(y_{ij}^2 + B^2\|x_{ij}\|^2) \notag\\
&= 2\sum_{j=1}^{n_i}\|z_{ij}\|^2 \cdot \sum_{j=1}^{n_i}(y_{ij}^2 + B^2\|x_{ij}\|^2). \label{eq:cauchy-schwarz-applied}
\end{align*}
Substituting:
\begin{align*}
\|\hat{u}_i(\Psi)\|^2 
&\le 2c_2^2\left[\frac{C_\kappa^2}{c_1^2}\|w_i\|^2 + 
\frac{2}{\sigma_{\min}^4}\left(\sum_{j=1}^{n_i}\|z_{ij}\|^2\right)\left(\sum_{j=1}^{n_i}(y_{ij}^2 + B^2\|x_{ij}\|^2)\right)\right] \notag\\
&\le C_1'\left[\|w_i\|^2 + \left(\sum_{j=1}^{n_i}\|z_{ij}\|^2\right)\left(\sum_{j=1}^{n_i}(y_{ij}^2 + \|x_{ij}\|^2)\right)\right], 
\end{align*}
where $C_1' := \max\{2c_2^2 C_\kappa^2/c_1^2, 4c_2^2 B^2/\sigma_{\min}^4, 4c_2^2/\sigma_{\min}^4\}$.

Consequently, we define a random variable:
$$
W_i^{(1)} := 1 + \|w_i\|^2 + \sum_{j=1}^{n_i}(|y_{ij}|^2 + \|x_{ij}\|^2 + \|z_{ij}\|^2).
$$
Since all terms are non-negative, the individual components are bounded by the envelope:
$$
\|w_i\|^2 \le W_i^{(1)}, \quad \sum_{j=1}^{n_i} \|z_{ij}\|^2 \le W_i^{(1)}, \quad \text{and} \quad \sum_{j=1}^{n_i} (|y_{ij}|^2 + \|x_{ij}\|^2) \le W_i^{(1)}.
$$
Applying these bounds to the inequality derived above:
\begin{align*}
\|\hat{u}_i(\Psi)\|^2 &\le C_1'\left[\|w_i\|^2 + \sum_{j=1}^{n_i}\|z_{ij}\|^2\cdot\sum_{j=1}^{n_i}(|y_{ij}|^2 + \|x_{ij}\|^2)\right] \\
&\le C_1' \left[ W_i^{(1)} + W_i^{(1)} \cdot W_i^{(1)} \right] \\
&= C_1' \left[ W_i^{(1)} + (W_i^{(1)})^2 \right].
\end{align*}
Since $W_i^{(1)} \ge 1$, we have $W_i^{(1)} \le (W_i^{(1)})^2$. Therefore,
$$
\|\hat{u}_i(\Psi)\|^2 \le C_1' \left[ (W_i^{(1)})^2 + (W_i^{(1)})^2 \right] = 2C_1' (W_i^{(1)})^2.
$$
Let $C_1:= 2C_1'$. We conclude that
$$
\sup_{\Psi \in \Theta} \|\hat{u}_i(\Psi)\|^2 \le C_1 (W_i^{(1)})^2.
$$

Regarding integrability, since $W_i^{(1)}$ consists of linear terms of the data, its square $(W_i^{(1)})^2$ consists of squared terms (e.g., $\|w_i\|^4, y_{ij}^4, \|x_{ij}\|^4$) and cross-products. Condition (A1) guarantees finite second moments for the response and covariates, and by the Generalized Hölder's Inequality, the expectation of the cross-products is finite provided the individual fourth moments are finite. Thus, $\mathbb{E}[(W_i^{(1)})^2] < \infty$.

In summary, under Conditions (A1), (A2) and (A3), there exists a constant $C_1 > 0$ such that
\begin{equation} \label{th1:eq2}
\sup_{\Psi \in \Theta} \|\hat{u}_i(\Psi)\|^2 \le C_1 (W_i^{(1)})^2,
\end{equation}
where $W_i^{(1)}$ is a random variable defined by linear terms, satisfying $\mathbb{E}(W_i^{(1)}) < \infty$.
Recall that
$$
h_i(u_i; \Psi) = \log\varphi(u_i; \kappa w_i, \Sigma) + 
\sum_{j=1}^{n_i} \log\left[\sum_{k=1}^K \pi_k(x_{ij}, z_{ij}; \alpha) 
\varphi(y_{ij}; x_{ij}^\top\beta_k + z_{ij}^\top u_i, \sigma_k^2)\right].
$$

We bound the random effects term and the likelihood term separately.
Recall the expression for the log-density of the Gaussian random effects:
$$
\log \varphi(u_i ; \kappa w_i, \Sigma)=-\frac{q}{2} \log (2 \pi)-\frac{1}{2} \log |\Sigma|-\frac{1}{2}(u_i-\kappa w_i)^{\top} \Sigma^{-1}(u_i-\kappa w_i) .
$$
We bound the three terms on the right-hand side uniformly over $\Theta$, using the eigenvalue bounds from Condition (A3), 
$\lambda_{\max}(\Sigma^{-1}) = 1/\lambda_{\min}(\Sigma) \le 1/c_1$, 
$$
(u_i-\kappa w_i)^{\top} \Sigma^{-1}(u_i-\kappa w_i) \leq \lambda_{\max }\left(\Sigma^{-1}\right)\|u-\kappa w_i\|^2 \leq c_1^{-1}\left(\|u\|+\|\kappa\|\|w\|\right)^2 \leq 2 c_1^{-1}\left(\|u\|^2+C_\kappa^2\|w\|^2\right).
$$
Therefore,
$$
|\log\varphi(u_i; \kappa w_i, \Sigma)| \le \frac{q}{2}\log(2\pi) + \frac{1}{2}|\log|\Sigma|| + 
\frac{1}{2} \cdot \frac{2}{c_1}(\|u\|^2 + C_\kappa^2\|w_i\|^2) \le C_2(1 + \|u_i\|^2 + \|w_i\|^2),
$$
where $C_2:= \max\{q\log(2\pi)/2, {1}/{2}\log(c_2^q), {1}/{c_1}, {C_\kappa^2}/{c_1}\}$. 
Substituting $u_i = \hat{u}_i(\Psi)$  and using equation (\ref{th1:eq2})
$$
\sup _{\Psi \subset \Theta}|\log \varphi(\hat{u}_i(\Psi) ; \kappa w_i, \Sigma)| \leq C_2\left(1+\left(C_1 W_i^{(1)}\right)^2+\|w_i\|^2\right) \leq C_3\left(1+\left(W_i^{(1)}\right)^2+\|w_i\|^2\right) .
$$


For the $j$-th observation of the $i$-th subject ($j=1, \dots, n_i$), define the contribution to the likelihood as:
$$
\ell_{ij}(\Psi) := \log \left[\sum_{k=1}^K \pi_k(x_{ij}, z_{ij}; \alpha) \varphi\left(y_{ij} ; x_{ij}^{\top} \beta_k+z_{ij}^{\top} \hat{u}_i(\Psi), \sigma_k^2\right)\right].
$$
Since $\sum_{k=1}^K \pi_k=1$ and  $\varphi(\cdot;\cdot,\cdot)$ is a density function, we have
$\sum_{k=1}^K \pi_k \varphi(y_{ij} ; x_{ij}^{\top} \beta_k+z_{ij}^{\top} \hat{u}_i(\Psi), \sigma_k^2) \leq\max _k \varphi(y_{ij} ; x_{ij}^{\top} \beta_k+z_{ij}^{\top} \hat{u}_i(\Psi), \sigma_k^2).$
The Gaussian density is bounded  by $\left(2 \pi \sigma_{\text {min }}^2\right)^{-1 / 2}$. Hence
$$
\ell_{ij}(\Psi)\leq \log \left(K \max _k \varphi(y_{ij} ; x_{ij}^{\top} \beta_k+z_{ij}^{\top} \hat{u}_i(\Psi), \sigma_k^2)\right) \leq \log K-\frac{1}{2} \log \left(2 \pi \sigma_{\min }^2\right)=: C^{+},
$$
a deterministic constant. 

For any fixed component $k$,
$$
\varphi\left(y_{ij} ; x_{ij}^{\top} \beta_k+z_{ij}^{\top} \hat{u}_i(\Psi), \sigma_k^2\right)=\frac{1}{\sqrt{2 \pi \sigma_k^2}} \exp \left(-\frac{\left(y_{ij}-x_{ij}^{\top} \beta_k+z_{ij}^{\top} \hat{u}_i(\Psi)\right)^2}{2 \sigma_k^2}\right).
$$
Thus
$$
\log \varphi\left(y_{ij} ; x_{ij}^{\top} \beta_k+z_{ij}^{\top} \hat{u}_i(\Psi), \sigma_k^2\right)=-\frac{1}{2} \log \left(2 \pi \sigma_k^2\right)-\frac{\left(y_{ij}-x_{ij}^{\top} \beta_k+z_{ij}^{\top} \hat{u}_i(\Psi)\right)^2}{2 \sigma_k^2}.
$$
Due to $\sigma_{ \min }^2\leq \sigma_k^2 \leq \sigma_{\max }^2$, $(y_{ij} - x_{ij}^\top \beta_k - z_{ij}^\top \hat{u}_i(\Psi))^2 \le 2y_{ij}^2+ 2(x_{ij}^\top \beta_k - z_{ij}^\top \hat{u}_i(\Psi))^2$,
then
$$
\log \varphi\left(y_{ij} ; x_{ij}^{\top} \beta_k+z_{ij}^{\top} \hat{u}_i(\Psi), \sigma_k^2\right) \geq-\frac{1}{2}\log(2\pi\sigma_{\max}^2)-\frac{1}{\sigma_{\min }^2}\left(y_{ij}^2 + (x_{ij}^\top \beta_k - z_{ij}^\top \hat{u}_i(\Psi))^2\right).
$$
Using the inequality $(a+b+c)^2 \le 3(a^2+b^2+c^2)$:
$$
(y_{ij} - x_{ij}^\top \beta_k - z_{ij}^\top \hat{u}_i(\Psi))^2 \le 3|y_{ij}|^2 + 3\|x_{ij}\|^2 \|\beta_k\|^2 + 3\|z_{ij}\|^2 \|\hat{u}_i(\Psi)\|^2.
$$
Substituting the bound for $\|\hat{u}_i(\Psi)\|$ derived  in equation \ref{th1:eq2}, i.e., $\|\hat{u}_i(\Psi)\|^2 \le C_1 (W_i^{(1)})^2$:
$$
(y_{ij} - x_{ij}^\top\beta_k - z_{ij}^\top \hat{u}_i(\Psi))^2 \le 3[|y_{ij}|^2 + B^2\|x_{ij}\|^2 + \|z_{ij}\|^2\|\hat{u}_i(\Psi)\|^2], 
$$

Therefore,
$$
\log \varphi\left(y_{ij} ;x_{ij}^{\top} \beta_k+z_{ij}^{\top} \hat{u}_i(\Psi), \sigma_k^2\right) 
 \geq-\frac{1}{2}\log(2\pi\sigma_{\max}^2) - \frac{3}{2\sigma_{\min}^2}[|y_{ij}|^2 + B^2\|x_{ij}\|^2 + \|z_{ij}\|^2\|\hat{u}_i(\Psi)\|^2].
$$
Since the mixture includes at least one component with a positive weight 
\begin{align*}
\ell_{ij}(\Psi)&=\log \sum_{k=1}^K \pi_k \varphi_k \geq \log \left(\min _k \pi_k \cdot \max _k \varphi_k\right) = \log\pi_{\min} + \max_k \log\varphi(y_{ij}; x_{ij}^\top\beta_k + z_{ij}^\top \hat{u}_i(\Psi), \sigma_k^2) \\
&\ge -C_3 - \frac{3}{2\sigma_{\min}^2}[|y_{ij}|^2 + B^2\|x_{ij}\|^2 + \|z_{ij}\|^2\|\hat{u}_i(\Psi)\|^2],
\end{align*}
where $C_3 := {1}/{2}\log(2\pi\sigma_{\max}^2) + |\log\pi_{\min}|$.

The mixture is at least one component times its weight, and the weights are positive and uniformly bounded below on a compact set $\Theta$ for fixed data.  Combining constants, we derive
$$
|\ell_{ij}(\Psi)| \le C_4\left(1+|y_{ij}|^2+\left\|x_{ij}\right\|^2+\left\|z_{ij}\right\|^2 \|\hat{u}_i(\Psi)\|^2\right),
$$
where $C_4 := \max\{C^+, C_3, {3B^2}/{2\sigma_{\min}^2}, {3}/{2\sigma_{\min}^2}\}$. 

Summing over all $n_i$ observations:
\begin{align*}
\left|\sum_{j=1}^{n_i} \ell_{ij}(\Psi)\right| 
&\le \sum_{j=1}^{n_i} |\ell_{ij}( \Psi)| 
\le C_4 \sum_{j=1}^{n_i} (1 + |y_{ij}|^2 + \|x_{ij}\|^2 + \|z_{ij}\|^2\|\hat{u}_i(\Psi)\|^2) \\
&= C_4\left[n_i + \sum_{j=1}^{n_i}(|y_{ij}|^2 + \|x_{ij}\|^2) + 
\|\hat{u}_i(\Psi)\|^2\sum_{j=1}^{n_i}\|z_{ij}\|^2\right],
\end{align*}
thus
$$
\sup_{\Psi \in \Theta}\left|\sum_{j=1}^{n_i} \ell_{ij}(\hat{u}_i(\Psi); \Psi)\right| 
\le C_4\left[n_i + \sum_{j=1}^{n_i}(|y_{ij}|^2 + \|x_{ij}\|^2) + 
C_1(W_i^{(1)})^2\sum_{j=1}^{n_i}\|z_{ij}\|^2\right]. 
$$
Combining the bounds 
\begin{align*}
\sup_{\Psi \in \Theta}|h_i(\hat{u}_i(\Psi); \Psi)| 
&\le \sup_{\Psi}|\log\varphi(\hat{u}_i(\Psi); \kappa w_i, \Sigma)| + 
\sup_{\Psi}\left|\sum_{j=1}^{n_i}\ell_{ij}(\hat{u}_i(\Psi); \Psi)\right| \\
&\le C_2(1 + C_1^2(W_i^{(1)})^2 + \|w_i\|^2) \\
&\quad + C_4\left[n_i + \sum_{j=1}^{n_i}(|y_{ij}|^2 + \|x_{ij}\|^2) + 
C_1(W_i^{(1)})^2\sum_{j=1}^{n_i}\|z_{ij}\|^2\right].
\end{align*}
The term $C_1(W_i^{(1)})^2 \sum_j \|z_{ij}\|^2$ can be bounded using
$$C_1(W_i^{(1)})^2 \sum_{j=1}^{n_i}\|z_{ij}\|^2 \le C_1(W_i^{(1)})^2 \cdot W_i^{(1)} = C_1(W_i^{(1)})^3.$$
Since $W_i^{(1)} \ge 1$, we have
$$
1 \le W_i^{(1)} \le (W_i^{(1)})^2 \le (W_i^{(1)})^3 .
$$

We now simplify the bound by defining an appropriate random variable $W_i^{(2)}$.
Define
$$
W_i^{(2)}:=n_i+\left(W_i^{(1)}\right)^3+\sum_{j=1}^{n_i}\left(|y_{ij}|^2+\left\|x_{ij}\right\|^2+\left\|z_{ij}\right\|^2\right)+\|w_i\|^2 .
$$
Note that by the definition of $W_i^{(1)}$,
$$
W_i^{(2)} = n_i + (W_i^{(1)})^3 + [W_i^{(1)} - 1].
$$

Therefore: 
\begin{align*}
\sup_{\Psi}|h_i(\hat{u}_i(\Psi); \Psi)| 
&\le C_2(1 + C_1(W_i^{(1)})^2 + \|w_i\|^2) + 
C_4(n_i + W_i^{(1)} + C_1(W_i^{(1)})^3) \\
&\le [C_2 + C_4] W_i^{(2)} + 
[C_2C_1 + C_4C_1] W_i^{(2)}  = C_5 \cdot W_i^{(2)},
\end{align*}
where $C_5 := (C_2 + C_4)(2 + C_1).$  

To establish that the envelope is integrable, i.e., $\mathbb{E}[W_i^{(2)}] < \infty$, we must examine the moments of the dominating terms.
Since $u_i \sim \mathcal{N}(\kappa w_i, \Sigma)$
and $\|a+b\|^6 \le 2^5(\|a\|^6 + \|b\|^6)$,
$$
\mathbb{E}[\|u_i\|^6] \le 32\|\kappa\|^6 \mathbb{E}[\|w_i\|^6] + 32\mathbb{E}[\|\tilde{u}_i\|^6],
$$
where $\tilde{u}_i \sim \mathcal{N}(0, \Sigma)$. Since $\tilde{u}_i$ is zero-mean Gaussian, all its moments are finite. 
Under Condition (A1), we have $\mathbb{E}[\|u_i\|^6] < \infty$.

Conditionally on covariates and random effects,
$$
y_{ij} = x_{ij}^\top\beta_k + z_{ij}^\top u_i + \epsilon_{ij},
$$
we have
$$
\mathbb{E}[|y_{ij}|^6] \le C \left( \mathbb{E}[\|x_{ij}\|^6 \|\beta_k\|^6] + \mathbb{E}[\|z_{ij}\|^6 \|u_i\|^6] + \mathbb{E}[\epsilon_{ij}^6] \right)<+\infty.
$$
All cross-product terms can be bounded using Hölder's inequality,
$$
\mathbb{E}[\|w_i\|^4 |y_{ij}|^2] \le \sqrt{\mathbb{E}[\|w_i\|^8]} \cdot \sqrt{\mathbb{E}[|y_{ij}|^4]} \le \mathbb{E}[\|w_i\|^6]^{2/3} \cdot \mathbb{E}[|y_{ij}|^6]^{1/3} < \infty.
$$
Therefore, $(W_i^{(1)})^3$ has a finite expectation under the strengthened moment conditions,
which implies
$$\mathbb{E}[W_i^{(2)}] < \infty. $$
Thus, there exists a constant $C_6 > 0$ and a random variable $W_i^{(2)}$ such that
$$
\sup_{\Psi \in \Theta} |h_i(\hat{u}_i, \Psi)| \le C_6 W_i^{(2)} 
$$
with  $E[W_i^{(2)}] < \infty.$
The Hessian is given by
$
H_i(\Psi)=\Sigma^{-1}+\sum_{j=1}^n \sum_{k=1}^K \gamma_{j k}(\Psi) {z_{ij} z_{ij}^{\top}}/{\sigma_k^2},
$
where $\gamma_{ijk}(\Psi) \in (0,1)$ and $\sum_{k=1}^K \gamma_{ijk}(\Psi) = 1$.
Since $H_i(\Psi)$  is positive definite, by Condition (A3), the eigenvalues are bounded  by
$$
\lambda_{\min }\left(H_i(\Psi)\right) \geq \lambda_{\min }\left(\Sigma^{-1}\right)=1 / \lambda_{\max }(\Sigma) \geq \frac{1}{c_2} :=\underline{\lambda}>0.
$$
By the sub-additivity of the maximum eigenvalue for positive definite matrices:
\begin{align*}
\lambda_{\max}(H_i(\Psi)) 
&\le \lambda_{\max}(\Sigma^{-1}) + \lambda_{\max}\left(\sum_{j=1}^{n_i}\sum_{k=1}^K 
\gamma_{ijk}(\Psi) \frac{z_{ij}z_{ij}^\top}{\sigma_k^2}\right) \notag\\
&\le \lambda_{\max}(\Sigma^{-1}) + \sum_{j=1}^{n_i}\sum_{k=1}^K \gamma_{ijk}(\Psi) 
\frac{\lambda_{\max}(z_{ij}z_{ij}^\top)}{\sigma_k^2} \notag\\
&= \lambda_{\max}(\Sigma^{-1}) + \sum_{j=1}^{n_i}\sum_{k=1}^K \gamma_{ijk}(\Psi) 
\frac{\|z_{ij}\|^2}{\sigma_k^2}.
\end{align*}
Since $\sum_{k=1}^K \gamma_{ijk}(\Psi) = 1$ and $\sigma_k^2 \ge \sigma_{\min}^2$,
$$
\lambda_{\max }\left(H_i(\Psi)\right) \leq \lambda_{\max }\left(\Sigma^{-1}\right)+\sum_{j, k} \gamma_{j k}(\Psi) \frac{\lambda_{\max }\left(z_{ij} z_{ij}^{\top}\right)}{\sigma_k^2} \leq c_1^{-1}+\sigma_{\min }^2 \sum_{j=1}^{n_i}\left\|z_{ij}\right\|^2:=\bar{\lambda}_i
$$
Thus, the eigenvalues of $H_i(\Psi)$ satisfy
$$
\underline{\lambda} \leq \lambda_r\left(H_i(\Psi)\right) \leq \bar{\lambda}_i, \quad r=1, \cdots, q.
$$
Consequently, since the determinant is the product of eigenvalues,
$$
\left|H_i(\Psi)\right|=\prod_{r=1}^q \lambda_r\left(H_i(\Psi)\right) \in\left[\underline{\lambda}^q, \bar{\lambda}_i^q\right].
$$
Since $|H_i(\Psi)| > 0$, we have 
$$
|\log | H_i(\Psi) \| \leq \max \left\{q|\log \underline{\lambda}|, q\left|\log \bar{\lambda}_i\right|\right\}.
$$
For the upper bound term,
$$
\bar{\lambda}_i = \frac{1}{c_1} + \frac{1}{\sigma_{\min}^2}\sum_{j=1}^{n_i}\|z_{ij}\|^2 
= \frac{1}{c_1}\left(1 + \frac{c_1}{\sigma_{\min}^2}\sum_{j=1}^{n_i}\|z_{ij}\|^2\right).
$$
Therefore,
$$
q\log\bar{\lambda}_i  = q\log\left[\frac{1}{c_1}\left(1 + \frac{c_1}{\sigma_{\min}^2}\sum_{j=1}^{n_i}\|z_{ij}\|^2\right)\right] = -q\log c_1 + q\log\left(1 + \frac{c_1}{\sigma_{\min}^2}\sum_{j=1}^{n_i}\|z_{ij}\|^2\right).
$$
Using the elementary inequality $\log (1+t) \leq t$ for $t \geq 0$,
$$
q\log\bar{\lambda}_i \le q|\log c_1| + q\cdot \frac{c_1}{\sigma_{\min}^2}\sum_{j=1}^{n_i}\|z_{ij}\|^2.
$$
Thus,
$$
|\log|H_i(\Psi)|| \le \max\left\{q\log c_2, q|\log c_1| + \frac{qc_1}{\sigma_{\min}^2}\sum_{j=1}^{n_i}\|z_{ij}\|^2\right\} \le C_6\left(1 + \sum_{j=1}^{n_i}\|z_{ij}\|^2\right),
$$
where $C_6 := \max\left\{q\log c_2, q|\log c_1|,{qc_1}/{\sigma_{\min}^2}\right\}.$

Define $W_i^{(3)} := 1 + \sum_{j=1}^{n_i}\|z_{ij}\|^2$.
We can derive 
$$
\sup_{\Psi \in \Theta}|\log|H_i(\Psi)|| \le C_6 W_i^{(3)}.
$$
By Condition (A1),
$$
\mathbb{E}[W_i^{(3)}] = 1 + \mathbb{E} [\sum_{j=1}^{n_i}\|z_{ij}\|^2] 
 < \infty.
$$
There exists a constant $C_6 > 0$ such that
$$
\sup_{\Psi \in \Theta} |\log|H_i(\Psi)|| \le C_6 W_i^{(3)}, 
$$
with  $\mathbb{E}[W_i^{(3)}] < \infty$.

Combine the bounds from the above proof to construct an integrable envelope for $\{\ell_{\text{LA}, i}(\Psi): \Psi \in \Theta\}$.
Recall the Laplace-approximated log-likelihood
$$
\ell_{\text{LA},i}(\Psi) = h_i(\hat{u}(\Psi), \Psi) + \frac{q}{2} \log(2\pi) - \frac{1}{2} \log |H_i(\Psi)|.
$$
Using the triangle inequality,
$$
|\ell_{\text{LA},i}(\Psi)| \le |h_i(\hat{u}_i(\Psi); \Psi)| + \frac{q}{2}|\log(2\pi)| + 
\frac{1}{2}|\log|H_i(\Psi)||. 
$$
Taking supremum over $\Psi \in \Theta$,
$$
\sup_{\Psi \in \Theta} |\ell_{\text{LA},i}(\Psi)| \leq \sup_{\Psi} |h_i(\hat{u}(\Psi), \Psi)| + \frac{q}{2} \log(2\pi) + \frac{1}{2} |\sup_{\Psi} |\log| H_i(\Psi)||.
$$
According to
$\sup_{\Psi \in \Theta}|h_i(\hat{u}_i(\Psi); \Psi)| \le C_6 W_i^{(2)}$
and 
$\sup_{\Psi \in \Theta}|\log|H_i(\Psi)|| \le C_7 W_i^{(3)}$,
we obtain
$$
\sup_{\Psi \in \Theta}|\ell_{\text{LA},i}(\Psi)| 
\le C_6 W_i^{(2)} + \frac{q}{2}\log(2\pi) + \frac{1}{2}C_7 W_i^{(3)}.
$$
Define the integrable envelope
$$
G_i :=  C_6 W_i^{(2)} + \frac{1}{2} C_7 W_i^{(3)} + \frac{q}{2} |\log(2\pi)|.
$$
Then for all $\Psi \in \Theta$, $|\ell_{\text{LA},i}(\Psi)| \le \sup_{\Psi \in \Theta}|\ell_{\text{LA},i}(\Psi)| \le G_i$.
Because $W_i^{(2)}$ and $W_i^{(3)}$ have finite expectations,
$$
\mathbb{E}[G_i] \leq C_6 \mathbb{E}[W_i^{(2)}] + \frac{1}{2} C_7 \mathbb{E}[W_i^{(3)}] + \frac{q}{2} |\log(2\pi)| < \infty.
$$
This establishes that $G_i$ is an integrable envelope for the class 
$\{\ell_{\text{LA},i}(\Psi): \Psi \in \Theta\}$.

Since the parameter space $\Theta$ is compact, and  $\ell_{\mathrm{LA}, i}(\Psi)$ are continuous in $\Psi$ and dominated by an integrable envelope $G_i$, the Glivenko–Cantelli property holds. By the Glivenko-Cantelli theorem for continuous functions on compact spaces,
$$
\sup_{\Psi \in \Theta}\left|\frac{1}{N}\sum_{i=1}^N \ell_{\text{LA},i}(\Psi) - 
\mathbb{E}[\frac{1}{N}\sum_{i=1}^N \ell_{\text{LA},i}(\Psi)]\right| \xrightarrow{p} 0 \quad \text{as } N \to \infty.
$$
Equivalently, in terms of the objective functions,
$$
\sup_{\Psi \in \Theta}|M_N(\Psi) - M(\Psi)| \xrightarrow{p} 0 \quad \text{as } N \to \infty.
$$
This completes the proof of uniform convergence. 
\end{proof}


\paragraph{Proposition 2}

Under Conditions (A2) and (A4), the true parameter $\Psi_0$ is the unique maximizer of the limiting objective function $M(\Psi)$. Specifically, for any $\epsilon > 0$, there exists a constant $\Delta(\epsilon) > 0$ such that:
$$\sup_{\|\Psi-\Psi_0\|\ge \epsilon} M(\Psi) \le M(\Psi_0) - \Delta(\epsilon).$$

This means the peak of the population objective function at $\Psi_0$ is strictly higher than its value anywhere outside an $\epsilon$-neighborhood of $\Psi_0$

\begin{proof}
    
Let $\hat\Psi$ be any value that maximizes $M_N(\Psi)$. By definition, $M_N(\hat\Psi) \ge M_N(\Psi)$ for all $\Psi \in \Theta$. In particular, $M_N(\hat\Psi) \ge M_N(\Psi_0)$.

Fix an arbitrary $\epsilon > 0$. According to Proposition 1, the well-separated maximum condition guarantees the existence of a constant $\Delta(\epsilon) > 0$ such that $M(\Psi) \le M(\Psi_0) - \Delta(\epsilon)$ whenever $\|\Psi - \Psi_0\| \ge \epsilon$. 
The Uniform Law of Large Numbers states that for any $\delta > 0$, we can find a sample size $N_\delta$ such that for all $N \ge N_\delta$, the event ${M}_N = \{\sup_{\Psi\in\Theta}|M_N(\Psi)-M(\Psi)| < \Delta(\epsilon)/2\}$ occurs with a probability of at least $1-\delta$.

First, by applying the triangle inequality at the true parameter $\Psi_0$, we have $M_N(\Psi_0) > M(\Psi_0) - \Delta(\epsilon)/2$. Second, for any parameter value $\Psi$ outside the $\epsilon$-ball around $\Psi_0$ (i.e., where $\|\Psi - \Psi_0\| \ge \epsilon$), we have $M_N(\Psi) < M(\Psi) + \Delta(\epsilon)/2$. Using the separation property, this becomes 
$$
M_N(\Psi) < (M(\Psi_0) - \Delta(\epsilon)) + \Delta(\epsilon)/2 = M(\Psi_0) - \Delta(\epsilon)/2.
$$
Combining these two inequalities on the event $\mathcal{E}_N$ reveals a crucial relationship:
$$
\sup_{\|\Psi-\Psi_0\|\ge \epsilon} M_N(\Psi) < M(\Psi_0) - \frac{\Delta(\epsilon)}{2} < M_N(\Psi_0).
$$
This chain of inequalities shows that the value of the sample objective function anywhere outside the $\epsilon$-ball around $\Psi_0$ is strictly less than its value at $\Psi_0$. Since $\hat\Psi$ is the maximizer of $M_N(\Psi)$, it must achieve a value at least as large as $M_N(\Psi_0)$. Therefore, $\hat\Psi$ cannot possibly lie outside the $\epsilon$-ball. This means that on the event $\mathcal{E}_N$, the estimator $\hat\Psi$ must be within a distance $\epsilon$ of $\Psi_0$.
\end{proof}

From Proposition 1 and Proposition 2, this implies that  $\{\|\hat\Psi - \Psi_0\| \ge \epsilon\}$ is a subset of the complement of $\mathcal{E}_N$, such $\{\|\hat{\Psi}_N - \Psi_0\| \ge \epsilon\} \subseteq \mathcal{E}_N^c$. Consequently, for all $N \ge N_\delta$, its probability is bounded: 
$$
P(\|\hat\Psi - \Psi_0\| \ge \epsilon) \le P(\mathcal{E}_N^c) < \delta. 
$$
Since $\delta$ can be made arbitrarily small by choosing a large enough $N$, we have shown that for any $\epsilon > 0$, $\lim_{N\to\infty} P(\|\hat\Psi - \Psi_0\| \ge \epsilon) = 0$. 

This is the definition of convergence in probability, which completes the proof of consistency.

\section*{Proof of Theorem 2 }

\subsection{ Proof of Lemma 1}
To prove Theorem 2, we provide the following Lemma.
\begin{lemma}\label{ui}
For subject $i$, let the random effects mode be $\hat{u}_i := \arg\max_u \ell_i(u)$. Under the conditions (B1)-(B3), When the number of observations $n_i \to \infty$, for subject $i$,
$$\|\hat{u}_i - u_i\| = O_p(n_i^{-1/2}).$$
\end{lemma}
\begin{proof}
Let $\hat{u}_i$ be the maximizer of  $h_i(u)$:
$$\hat{u}_i = \arg\max_{u \in \mathbb{R}^q} h_i(u).$$
By definition, the  mode $\hat{u}_i$ is a stationary point of the log-posterior function; it must satisfy the first-order optimality condition:$\nabla_u h_i(\hat{u}_i)  =0.$
The second-order Taylor expansion of this score vector around the true value $u_i$:
$$
0 = \nabla_u \ell_i(\hat{u}_i) = \nabla_u \ell_i(u_i) + \nabla_u^2 \ell_i(\bar{u}_i) (\hat{u}_i - u_i),
$$
where $\bar{u}_i$ is a point on the line segment between $\hat{u}_i$ and $u_i$. Let $s_i(u_i):= \nabla_u \ell_i(u_i)$ denote the score vector evaluated at the true value, and let $H_i(u):= -\nabla_u^2 \ell_i(u)$ be the negative Hessian matrix. The expansion can be rewritten as:
$$0 = s_i(u_i) - H_i(\bar{u}_i) (\hat{u}_i - u_i).$$
Rearranging this equation gives a fundamental expression for the estimation error:
$$\hat{u}_i - u_i = H_i(\bar{u}_i)^{-1} s_i(u_i).$$

Under the assumption of a correctly specified model, the conditional expectation of the score, given the true random effect $u_i$ and covariates, is zero. The variance of the score, being a sum of approximately independent, mean-zero terms, scales linearly with the number of observations. Thus, we have:
$$
\mathbb{E}[s_i(u_i) | u_i, x_{ij}, z_{ij}] = 0 \quad \text{and} \quad \text{Var}(s_i(u_i)) = O(n_i).
$$
According to Chebyshev's inequality, a mean-zero random vector with variance of order $O(n_i)$ has a stochastic magnitude of order $O_p(\sqrt{n_i})$. Therefore,  $\|s_i(u_i)\| = O_p(n_i^{1/2}).$ 

First, Condition (B1) directly provides the key property of the Hessian: its eigenvalues grow linearly with $n_i$. The norm of the inverse of a symmetric matrix is the reciprocal of its smallest eigenvalue. There exists a constant$c_h>0$ such that$\lambda_{\min}\{H_i(u)\} \ge c_h n_i$, so we have:
$\|H_i(u)^{-1}\| \le \frac{1}{c_h n_i}.$
So the inverse Hessian shrinks at a rate of $O(n_i^{-1})$ for any $u$ in the specified neighborhood.

Assuming for a moment that $\bar{u}_i$ is in the neighborhood, we can bound the error:
$$\|\hat{u}_i - u_i\| = \|H_i(\bar{u}_i)^{-1} s_i(u_i)\| \le \|H_i(\bar{u}_i)^{-1}\| \cdot \|s_i(u_i)\| = O_p(n_i^{-1}) \cdot O_p(n_i^{1/2}) = O_p(n_i^{-1/2}).
$$
This preliminary result shows that the error $\|\hat{u}_i - u_i\|$ converges to zero in probability as $n_i \to \infty$. Since $\bar{u}_i$ lies between $\hat{u}_i$ and $u_i$, it must also converge to $u_i$. Therefore, for any fixed neighborhood around $u_i$, $\bar{u}_i$ will eventually lie within it with probability approaching one. This validates the use of the Hessian bound at the point $\bar{u}_i$.
\end{proof}

\subsection{Proof of Theorem 2}
\begin{proof}

Suppose Conditions (A1)--(A5) and (B1)--(B4) hold, and let $\Psi_0=(\beta_0^\top,\theta_0^\top)^\top$ denote the true parameter, where $\beta_0=(\beta_{1,0}^\top,\dots,\beta_{K,0}^\top)^\top$. By Lemma \ref{ui}, we can replace 
$\hat{u}_i(\Psi_0)$ with $u_i$ in the score, enabling application of 
the central limit theorem.

For a fixed $k\in\{1,\dots,K\}$, define the subject-level score and Hessian blocks
$$
U_{i,\beta_k}(\Psi):=\frac{\partial   \ell_{{\rm LA},i}(\Psi)}{\partial\beta_k}
,
\qquad
H_{i,\beta_k}(\Psi):=
-\,\frac{\partial^2  \ell_{{\rm LA},i}(\Psi)}{\partial\beta_k\partial\beta_k^\top}.$$

Then $\mathbb E\big[H_{i,\beta_k}(\Psi_0)\big]^{-1} $ is positive definite and the Laplace–MLE $\hat\beta_k$ satisfies
$$
\sqrt{N}\,(\hat\beta_k-\beta_{k,0})
\ \xrightarrow{d}\
\mathcal N\big(0,\ V_{\beta_k}\big),
$$
where
$$
V_{\beta_k}=\mathbb E\big[H_{i,\beta_k}(\Psi_0)\big]^{-1} \mathbb E\big[U_{i,\beta_k}(\Psi_0)\,U_{i,\beta_k}(\Psi_0)^\top\big]\mathbb E\big[H_{i,\beta_k}(\Psi_0)\big]^{-1}.
$$

Our model's Laplace log likelihood is 
$$
\ell_{\rm LA}(\Psi) = \sum_{i=1}^N \left\{ h_i(\hat{u}_i; \Psi) + \frac{q}{2}\log(2\pi) - \frac{1}{2} \log|{H}_i| \right\}, 
$$
where
$$
h_i(\hat{u}_i; \Psi) = \log \varphi(\hat{u}_i; \kappa\omega_i, \Sigma) + \sum_{j=1}^{n_i} \log \left[ \sum_{k=1}^{K} \pi_k(x_{ij}, z_{ij}; \alpha) \varphi\left(y_{ij}; x_{ij}^\top\beta_k + z_{ij}^\top \hat{u}_i, \sigma_k^2\right) \right].
$$
Define 
$$
\gamma_{ij,k}(\Psi, \hat{u}_i) = \frac{\pi_k(x_{ij}, z_{ij}; \alpha) \varphi_{ij,k}}{\sum_{m=1}^{K} \pi_m(x_{ij}, z_{ij}; \alpha) \varphi_{ij,m}}.
$$
We obtain 
\begin{align*}
\left.\frac{\partial \ell_{\rm LA}}{\partial \beta_k}\right|_{u=\hat u} &=\sum_{i=1}^N \left.\frac{\partial h_i}{\partial \beta_k}\right|_{u=\hat u} = \sum_{i=1}^N \left( \sum_{j=1}^{n_i} \frac{\pi_k(x_{ij}, z_{ij}; \alpha) \varphi_{ij,k}}{\sum_{m=1}^{K} \pi_m(x_{ij}, z_{ij}; \alpha) \varphi_{ij,m}} \frac{y_{ij} -  x_{ij}^\top\beta_k + z_{ij}^\top \hat{u}_i }{\sigma_k^2} x_{ij} \right)  \\
&= \sum_{i=1}^N  \sum_{j=1}^{n_i} \gamma_{ij,k} \frac{y_{ij} -  x_{ij}^\top\beta_k + z_{ij}^\top \hat{u}_i }{\sigma_k^2} x_{ij}.
\end{align*}
According to
$$
\frac{\partial^2 \mathcal{L}}{\partial \beta \partial \beta} = \frac{\partial^2 h}{\partial \beta \partial \beta} - \frac{\partial^2 h}{\partial \beta \partial \hat{u}} \left( \frac{\partial^2 h}{\partial \hat{u} \partial \hat{u}^T} \right)^{-1} \frac{\partial^2 h}{\partial \hat{u} \partial \beta}, 
$$ 
 the second derivative of $\ell_{\rm LA}$ is
\begin{align*}
\frac{\partial^2 \ell_{\rm LA}}{\partial \beta_k \partial \beta_k^\top}
&= \sum_{i=1}^{N}\sum_{j=1}^{n_i}
\Biggl[
-\frac{\gamma_{ijk}}{\sigma_k^2}\, x_{ij} x_{ij}^\top
+ \frac{\gamma_{ijk}(1-\gamma_{ijk})}{\sigma_k^4}
\Bigl(y_{ij} - x_{ij}^\top\beta_k - z_{ij}^\top \hat{u}_i\Bigr)^2
\, x_{ij} x_{ij}^\top
\notag\\
&\qquad\qquad
- \frac{\gamma_{ijk}}{\sigma_k^2}\,
x_{ij} z_{ij}^\top [H_i]^{-1}
\sum_{j'=1}^{n_i} \frac{\gamma_{ij'k}}{\sigma_k^2}\,
z_{ij'} x_{ij'}^\top
\Biggr],
\end{align*}
where 

\begin{equation*}
\centering
\begin{aligned}
\frac{\partial^2 h_i}{\partial \beta_k \partial \beta_k^\top}
&= -\sum_{j=1}^{n_i} \frac{x_{ij}x_{ij}^\top}{\sigma_k^2}
\Biggl[
\gamma_{ijk}
- \gamma_{ijk}(1-\gamma_{ijk})
\frac{\bigl(y_{ij}-x_{ij}^\top\beta_k-z_{ij}^\top \hat{u}_i\bigr)^2}{\sigma_k^2}
\Biggr],\\[0.9em]
\frac{\partial^2 h_i}{\partial u_i \partial u_i^\top}
&= -\Sigma^{-1}
+ \sum_{j=1}^{n_i} z_{ij}z_{ij}^\top
\Biggl\{
\Biggl[
\sum_{k=1}^K \gamma_{ijk}
\left(\frac{y_{ij}-x_{ij}^\top\beta_k-z_{ij}^\top \hat{u}_i}{\sigma_k^2}\right)^{\!2}\\
&\qquad\qquad
-\left(
\sum_{k=1}^K \gamma_{ijk}
\frac{y_{ij}-x_{ij}^\top\beta_k-z_{ij}^\top \hat{u}_i}{\sigma_k^2}
\right)^{\!2}
\Biggr]
-\sum_{k=1}^K \frac{\gamma_{ijk}}{\sigma_k^2}
\Biggr\},
\end{aligned}
\end{equation*}

\begin{equation*}
\centering
\begin{aligned}
\frac{\partial^2 h_i}{\partial \beta_k \partial u_i^\top}
&= \sum_{j=1}^{n_i} \frac{\gamma_{ijk}}{\sigma_k^2}
\Biggl[
\frac{\bigl(y_{ij}-x_{ij}^\top\beta_k-z_{ij}^\top \hat{u}_i\bigr)^2}{\sigma_k^2}\\
&\qquad
-\bigl(y_{ij}-x_{ij}^\top\beta_k-z_{ij}^\top \hat{u}_i\bigr)
\left(
\sum_{\ell=1}^K \frac{\gamma_{ij\ell}\bigl(y_{ij}-x_{ij}^\top\beta_\ell-z_{ij}^\top \hat{u}_i\bigr)}{\sigma_\ell^2}
\right)
-1
\Biggr]\,
x_{ij} z_{ij}^\top,\\[0.9em]
\frac{\partial^2 h_i}{\partial u_i \partial \beta_k^\top}
&= \sum_{j=1}^{n_i} \frac{\gamma_{ijk}}{\sigma_k^2}
\Biggl[
\frac{\bigl(y_{ij}-x_{ij}^\top\beta_k-z_{ij}^\top \hat{u}_i\bigr)^2}{\sigma_k^2}\\
&\qquad
-\bigl(y_{ij}-x_{ij}^\top\beta_k-z_{ij}^\top \hat{u}_i\bigr)
\left(
\sum_{\ell=1}^K \frac{\gamma_{ij\ell}\bigl(y_{ij}-x_{ij}^\top\beta_\ell-z_{ij}^\top \hat{u}_i\bigr)}{\sigma_\ell^2}
\right)
-1
\Biggr]\,
z_{ij} x_{ij}^\top.
\end{aligned}
\end{equation*}
We assume that
$$s_{ijk} := \gamma_{ijk}- 
\gamma_{ijk}\frac{(y_{ij} -  x_{ij}^\top\beta_k - z_{ij}^\top \hat{u}_i)^2}{\sigma_k^2} + \gamma_{ijk}(y_{ij} -  x_{ij}^\top\beta_k - z_{ij}^\top \hat{u}_i) \left( \sum_{k=1}^K \frac{\gamma_{ijk} (y_{ij} -  x_{ij}^\top\beta_k - z_{ij}^\top \hat{u}_i)}{\sigma_k^2} \right),$$
$$
t_{ijk} := \sum_{k=1}^K \sum_{l=1}^K \gamma_{ijk}  \gamma_{ijl}\left[\left(\frac{{y_{ij} -  x_{ij}^\top\beta_k - z_{ij}^\top \hat{u}_i }}{\sigma_k^2}\right) - \left(\frac{{y_{ij} -  x_{ij}^\top\beta_l - z_{ij}^\top \hat{u}_i }}{\sigma_l^2}\right)\right]^2,
$$
$$q_{ijk} :=\gamma_{ijk} - \gamma_{ijk}(1-\gamma_{ijk}) \frac{({y_{ij} -  x_{ij}^\top\beta_k - z_{ij}^\top \hat{u}_i})^2}{\sigma_k^2}.
$$ 
Therefore, we can rewrite them as follows:
\begin{eqnarray*}
\frac{\partial^2 h_i}{\partial \beta_k \partial \beta_k^\top} &=& -\sum_{j=1}^{n_i} q_{ijk}\frac{x_{ij}x_{ij}^\top}{\sigma_k^2},
\quad
\frac{\partial^2 h_i}{\partial u_i \partial u_i^\top}  = - \sum_{j=1}^{n_i} \sum_{k=1}^{K} t_{ijk}\frac{z_{ij}z_{ij}^\top}{\sigma_k^2},\\
\frac{\partial^2 h_i}{\partial \beta_k \partial u_i^\top}& =&  - \sum_{j=1}^{n_i} \sum_{k=1}^{K} s_{ijk}\frac{x_{ij}z_{ij}^\top}{\sigma_k^2}, \quad
\frac{\partial^2 h_i}{\partial u_i \partial \beta_k^\top} =  - \sum_{j=1}^{n_i} \sum_{k=1}^{K} s_{ijk}\frac{z_{ij}x_{ij}^\top}{\sigma_k^2}. 
\end{eqnarray*}
Apply the uniform law of large numbers, for any compact neighborhood $\mathcal{N}(\Psi_0) \subset \Theta$:
$$
\sup_{\Psi \in \mathcal{N}(\Psi_0)} \left\|\frac{1}{N}\sum_{i=1}^N H_{i,\beta_k}(\Psi) - E[H_{i,\beta_k}(\Psi_0)]\right\| \xrightarrow{p} 0.
$$
We assume $ r_{ijk} = y_{ij} -  x_{ij}^\top\beta_k - z_{ij}^\top \hat{u}_i$. Differentiate $q_{ijk}$ with respect to $\beta_k$:
\begin{align*}
\frac{\partial q_{ijk}}{\partial \beta_k} &= \frac{\partial \gamma_{ijk}}{\partial \beta_k } - \frac{\partial \gamma_{ijk}}{\partial \beta_k } (1-\gamma_{ijk})\frac{{r_{ijk}}^2}{\sigma_k^2} - \gamma_{ijk}  \frac{ \partial (1-\gamma_{ijk} )}{\partial \beta_k } \frac{{r_{ijk}}^2}{\sigma_k^2}   - \gamma_{ijk}(1-\gamma_{ijk})  \frac{ \partial }{\partial \beta_k } \frac{{r_{ijk}}^2}{\sigma_k^2} \\
&= \frac{\partial \gamma_{ijk}}{\partial \beta_k} \left[ 1 - \frac{(1-\gamma_{ijk}){r_{ijk}}^2}{\sigma_k^2} \right] + \gamma_{ijk}\frac{\partial \gamma_{ijk}}{\partial \beta_k}\frac{{r_{ijk}}^2}{\sigma_k^2}  -  \frac{2\gamma_{ijk}(1-\gamma_{ijk})r_{ijk}}{\sigma_k^2} x_{ij} \\
& =   \frac{\partial \gamma_{ijk}}{\partial \beta_k} \left[ 1 - \frac{ (1- 2 \gamma_{ijk} )r_{ijk}^2}{\sigma_k^2} \right] - \frac{2\gamma_{ijk}(1-\gamma_{ijk})r_{ijk}}{\sigma_k^2} x_{ij},  
\end{align*}
and hence we can contribute
\begin{align*}
\left\| \frac{\partial q_{ijk}}{\partial \beta_k}  \right\| &\le  \left\| \frac{\partial \gamma_{ijk}}{\partial \beta_k}  \right\| \cdot \left| 1 - \frac{(1-2\gamma_{ijk})r_{ijk}^2}{\sigma_k^2} \right| +  \left\| \frac{2\gamma_{ijk}(1-\gamma_{ijk})r_{ijk}}{\sigma_k^2} x_{ij} \right\| 
\end{align*}
\begin{align*}
&\le  \left\| \frac{\partial \gamma_{ijk}}{\partial \beta_k}  \right\| \cdot \left[ 1 + \frac{1}{\sigma_k^2} \left( |\epsilon_{ijk}| + \|x_{ij}\|\|\beta_k\| + \|z_{ij}\|\|\hat{u}_i\| \right)^2 \right] \\ &+ \frac{1}{2\sigma_k^2} \left( |\epsilon_{ijk}| + \|x_{ij}\|\|\beta_k\| + \|z_{ij}\|\|\hat{u}_i\| \right) \|x_{ij}\|.
\end{align*}
Using the derivative property of Softmax,
\begin{eqnarray*}  
\frac{\partial \gamma_{ijk}}{\partial \beta_k} &=& \gamma_{ijk} \frac{x_{ij}r_{ijk}}{\sigma_k^2} - \gamma_{ijk} \sum_{\ell=1}^K \gamma_{ij\ell} \frac{x_{ij}r_{ij\ell}}{\sigma_\ell^2},\\
\left\|\frac{\partial \gamma_{ijk}}{\partial \beta_k}\right\| &\leq& \gamma_{ijk} \left( \left\| \frac{x_{ij}r_{ijk}}{\sigma_k^2} \right\| + \sum_{\ell=1}^K \gamma_{ij\ell} \left\| \frac{x_{ij}r_{ij\ell}}{\sigma_\ell^2} \right\| \right).
\end{eqnarray*}
Due to $\gamma_{ijk} \le 1$ and $\sigma_k^2 \ge \sigma_{\min}^2$, 
$$
\left\|\frac{\partial \gamma_{ijk}}{\partial \beta_k}\right\| \leq 1 \cdot \left( \frac{\|x_{ij}\|\max_\ell|r_{ij\ell}|}{\sigma_{\min}^2} + \left(\sum_{\ell}\gamma_{ij\ell}\right) \frac{\|x_{ij}\|\max_\ell|r_{ij\ell}|}{\sigma_{\min}^2} \right) = \frac{2\|x_{ij}\|\max_\ell|r_{ij\ell}|}{\sigma_{\min}^2}, 
$$ 
and then we derive 
$$ 
\left\|\frac{\partial \gamma_{ijk}}{\partial \beta_k}\right\|   \leq \frac{2\|x_{ij}\|}{\sigma_{\min}^2} [\max_k|\epsilon_{ijk}| + \|x_{ij}\|\|\beta_k\| + \|z_{ij}\|\|\hat{u}_i\|].
$$ 
Assume that $M_{ij} := |\max_k|\epsilon_{ijk}| + \|x_{ij}\|\|\beta_k\| + \|z_{ij}\|\|\hat{u}_i\|$, 
\begin{align*}
\left\| \frac{\partial q_{i j k}}{\partial \beta_{k}} \right\| &\leq \left\| \frac{\partial q_{ijk}}{\partial \beta_k}  \right\| \cdot[ 1 + \frac{1}{\sigma_k ^2}M_{ij} ^2 ] +  \frac{1}{2 \sigma_k ^2} M_{ij} \| x_{ij}\| \\
&  \leq  \frac{2\|x_{ij}\|}{\sigma_{\min}^2} M_{ij} + \frac{2\|x_{ij}\|}{\sigma_{\min}^2 \sigma_k^2} M_{ij}^3 + \frac{1}{2\sigma_k^2} M_{ij} \|x_{ij}\| \leq  \frac{1}{2 \sigma_{\min} ^2} B^2 \|x_{ij}\|^4.
\end{align*}
Under the stated assumptions, together with the normality of the noise and the finiteness of $K$,  the expectation of the preceding quantity is finite. This bounded gradient implies that $q_{ijk}$ is Lipschitz continuous with respect to $\beta_k$. Similarly, we can conclude that 
$s_{ijk}$ and $t_{ijk}$ is also Lipschitz continuous, with Lipschitz constant $O(2 M_{ij}^3{\|x_{ij}\|}/{\sigma_{\min}^4} )$. 

Next, we consider the second-order derivative term
$$
H_{i,\beta_k}(\Psi) = \frac{\partial^2 h_i(\hat{u}_i, \Psi)}{\partial \beta_k \partial \beta_k^\top} + \frac{\partial^2 h_i (\hat{u}_i, \Psi)}{\partial \beta_k \partial u_i^\top} \cdot [H_i(\hat{u}_i, \Psi)]^{-1} \cdot \frac{\partial^2 h_i(\hat{u}_i, \Psi)}{\partial u_i \partial \beta_k ^\top}.
$$
The difference between $H_{i,\beta_k}(\Psi)$ at $\Psi$ and $\tilde{\Psi}$ can be decomposed into four terms:
$$
H_{i,\beta_k}(\Psi) - H_{i,\beta_k}(\tilde{\Psi}) = \Delta_1 + \Delta_2 + \Delta_3 + \Delta_4.
$$
The first term represents the difference in the direct second derivatives with respect to $\beta_k$:
$$
\Delta_1 = \left[ \frac{\partial^2 h_i (\hat{u}_i, \Psi)}{\partial \beta_k \partial \beta_k^\top} - \frac{\partial^2 h_i (\hat{u}_i, \tilde{\Psi})}{\partial \beta_k \partial \beta_k^\top} \right].
$$
The remaining terms arise from the expansion of the product involving the inverse random-effect Hessian matrix $[H_i]^{-1}$. By adding and subtracting intermediate terms, we obtain
$$
\begin{aligned}
\Delta_2 &= \left[\frac{\partial^2 h_i (\hat{u}_i, \Psi)}{\partial \beta_k \partial u_i^\top} - \frac{\partial^2 h_i (\hat{u}_i, \tilde{\Psi})}{\partial \beta_k \partial u_i^\top}\right] [H_i(\hat{u}_i, \tilde{\Psi})]^{-1} \frac{\partial^2 h_i (\hat{u}_i, \tilde{\Psi})}{\partial u_i \partial \beta_k}, \\
\Delta_3 &= \frac{\partial^2 h_i (\hat{u}_i, \Psi)}{\partial \beta_k \partial u_i^\top} [H_i(\hat{u}_i, \tilde{\Psi})]^{-1} \left[\frac{\partial^2 h_i (\hat{u}_i, \Psi)}{\partial u_i \partial \beta_k} - \frac{\partial^2 h_i (\hat{u}_i, \tilde{\Psi})}{\partial u_i \partial \beta_k}\right], \\
\Delta_4 &= \frac{\partial^2 h_i (\hat{u}_i, \Psi)}{\partial \beta_k \partial u_i^\top} \left([H_i(\hat{u}_i, \Psi)]^{-1} - [H_i(\hat{u}_i, \tilde{\Psi})]^{-1}\right) \frac{\partial^2 h_i (\hat{u}_i, \Psi)}{\partial u_i \partial \beta_k}.
\end{aligned}
$$

Bound for $\Delta_1$: Recall that the Hessian (or its principal part) is given by
$$
\frac{\partial^2 h_i(\hat{u}_i; \Psi)}{\partial \beta_k \partial \beta_k^\top} = -\sum_{j=1}^{n_i} \frac{q_{ijk}}{\sigma_k^2} x_{ij} x_{ij}^\top.
$$
Using the Mean Value Theorem  to the scalar weights $q_{ijk}$, we have
\begin{equation}
    |q_{ijk}(\Psi) - q_{ijk}(\tilde{\Psi})| \le \|\nabla_\beta q_{ijk}\| \|\beta_k - \tilde{\beta_k}\|.
\label{eq:qijk_gradient_bound}
\end{equation}
Substituting the previously derived bound for $\|\nabla q_{ijk}\|$, we obtain
 $$\begin{aligned}
\|\Delta_1\| &\le \frac{1}{\sigma_{\min}^2} \sum_{j=1}^{n_i} \left( \frac{B^2}{\sigma_{\min}^2} \|x_{ij}\|^4 \|\beta_k - \tilde{\beta_k}\| \right) \|x_{ij}\|^2 \\
&= \left( \frac{B^2}{\sigma_{\min}^2} \sum_{j=1}^{n_i} \|x_{ij}\|^6 \right) \|\beta_k - \tilde{\beta_k}\|  .
\end{aligned}$$
This confirms that $\Delta_1$ scales linearly with the subject size $n_i$.
Define  
$$
B_i(\Psi) :=\|  \frac{\partial^2 h_i (\hat{u}_i, \Psi)}{\partial \beta_k \partial u_i^\top} \| = \left\| \sum_{j=1}^{n_i}  \frac{q_{ijk}}{\sigma_k^2}  x_{ij} z_{ij}^\top \right\|,
$$
we can get
\begin{eqnarray*}
\left\| \sum_{j=1}^{n_i} \frac{q_{ijk}(\tilde{\Psi})}{\sigma_k^2} x_{ij} z_{ij}^\top \right\| &\le& \sum_{j=1}^{n_i} \frac{|q_{ijk}|}{\sigma_{\min}^2} \|x_{ij}\| \|z_{ij}\| \le C_B \sum_{j=1}^{n_i} \|x_{ij}\| \|z_{ij}\|,\\
B_i(\Psi) - B_i(\tilde{\Psi}) &=& -\sum_{j=1}^{n_i} \frac{1}{\sigma_k^2} \left( q_{ijk}(\Psi) - q_{ijk}(\tilde{\Psi}) \right) x_{ij} z_{ij}^\top.
\end{eqnarray*}
By the Mean Value Theorem applied to the scalar function $q_{ijk}(\Psi)$, with the bound derived in equation \eqref{eq:qijk_gradient_bound}, and using the gradient bound derived earlier, substituting this back:
$$\begin{aligned}
\| B_i(\Psi) - B_i(\tilde{\Psi}) \| &\le \sum_{j=1}^{n_i} \frac{1}{\sigma_{\min}^2} \left( \frac{1}{2 \sigma_{\min} ^2} B^2 \|x_{ij}\|^4 \|\Psi - \tilde{\Psi}\| \right) \|x_{ij}\| \|z_{ij}\| \\
&\le \left( C_B' \sum_{j=1}^{n_i} \|x_{ij}\|^5 \|z_{ij}\| \right)\|\Psi - \tilde{\Psi}\|.
\end{aligned}$$

Decompose $\Delta_2$  into three components 
$$\|\Delta_2\| \le \| B_i(\Psi) - B_i(\tilde{\Psi}) \| \cdot \| [H_i(\tilde{\Psi})]^{-1} \| \cdot \| B_i(\tilde{\Psi}) \|,
$$
and 
$$  \Delta_ 3 \le \|{B_i(\Psi)} \| \cdot \| [H_i(\tilde{\Psi})] ^{-1} \| \cdot  \|  B_i(\Psi)^\top - B_i(\tilde{\Psi})^\top \|. 
$$
Bound for $\Delta_4$: According to $A^{-1} - B^{-1} = -A^{-1}(A - B)B^{-1}$, we expand the difference of inverses,
$$
\left [H_i(\hat{u}_i, \Psi)(\Psi)\right]^{-1}-\left[H_i(\hat{u}_i, \tilde{\Psi})\right]^{-1} = -\left[H_i(\hat{u}_i, \Psi)\right]^{-1}\left(H_i(\hat{u}_i, \Psi)-H_i(\hat{u}_i, \tilde{\Psi})\right)\left[H_i(\hat{u}_i, \tilde{\Psi})\right]^{-1},
$$  
and hence
$$
\Delta_4 = - \frac{\partial^2 h_i (\hat{u}_i, \Psi)}{\partial \beta_k \partial u_i^\top} \left(\left[H_i(\hat{u}_i, \Psi)\right]^{-1}\left(H_i(\hat{u}_i, \Psi)-H_i(\hat{u}_i, \tilde{\Psi})\right)\left[H_i(\hat{u}_i, \tilde{\Psi})\right]^{-1}\right) \frac{\partial^2 h_i (\hat{u}_i, \Psi)}{\partial u_i \partial \beta_k}.
$$
Because
$$ H_i(\hat{u}_i, \Psi) = \frac{\partial^2 h}{\partial \hat{u} \partial \hat{u}^{\top}} = -\Sigma^{-1} - \sum_{j=1}^{n_i} \sum_{k=1}^K q_{ijk}z_{ij}z_{ij}^\top,$$
we can derive
$$
\begin{aligned}
\mathbb{E}\left(\frac{\partial^2 h}{\partial \hat{u} \partial \hat{u}^{\top}}\right) &=  \mathbb{E}\left( \sum_{k=1}^K q_{ijk}z_{ij}z_{ij}^\top\right)  \geq \mathbb{E}\left( 2(k-1) \cdot \min_k( (x_{ij}^\top \beta_k+ z_{ij}^\top u_i - y_{ij}) ^2 +1 ) \cdot z_{ij}z_{ij}^\top) \right) \\
& \geq   \left[  \frac{(K -1)\min(\beta_k -\beta_l)^2 }{4 \sigma_{max}^2} +1\right] \mathbb{E}(z_{ij}z_{ij}^\top) \geq c I_q,
\end{aligned}
$$
where $I_q$ is a $q\times q$ dimensional identity matrix.
According to the Law of Large Numbers, when the sample size \(n_i\)is sufficiently large, $ H_i(\hat{u}_i, \Psi) $ converges in probability to its expectation. Thus, for sufficiently large \(n_i\), we have
$$\lambda_{\min}( H_i(\hat{u}_i, \Psi) ) \ge \frac{C}{2} = \lambda_1, \quad  \| [H_i(\hat{u}_i, \Psi)]^{-1} \| \le \frac{1}{\lambda_1}.
$$
 Therefore,
$$
\left\| \left[ H_i(\hat{u}_i, \Psi) \right]^{-1} \left( H_i(\hat{u}_i, \Psi) - H_i(\hat{u}_i, \tilde{\Psi}) \right) \left[ H_i(\hat{u}_i, \tilde{\Psi}) \right]^{-1} \right\| \leq \frac{1}{\lambda_1^2} \left\| H_i(\hat{u}_i, \Psi) - H_i(\hat{u}_i, \tilde{\Psi}) \right\|.
$$
Recall the structure of the Hessian matrix with respect to $u_i$,  and the difference is driven entirely by the weights $q_{ijk}$:
$$
H_i(\Psi) - H_i(\tilde{\Psi}) = \sum_{j=1}^{n_i} \sum_{k=1}^K (q_{ijk}(\Psi) - q_{ijk}(\tilde{\Psi})) z_{ij} z_{ij}^\top.
$$
Applying the gradient bound derived previously:
$$
|q_{ijk}(\Psi) - q_{ijk}(\tilde{\Psi})| \le \|\nabla_\beta q_{ijk}\| \|\Psi - \tilde{\Psi}\| \le \left( \frac{B^2}{\sigma_{\min}^2} \|x_{ij}\|^4 \right) \|\Psi - \tilde{\Psi}\|.
$$
Substituting this back:
$$\begin{aligned}
\| H_i(\Psi) - H_i(\tilde{\Psi}) \| &\le \sum_{j=1}^{n_i} \sum_{k=1}^K |q_{ijk}(\Psi) - q_{ijk}(\tilde{\Psi})| \cdot \|z_{ij} z_{ij}^\top\| \\
&\le \left( \sum_{j=1}^{n_i} \frac{K B^2}{\sigma_{\min}^2} \|x_{ij}\|^4 \|z_{ij}\|^2 \right) \|\Psi - \tilde{\Psi}\| \le  \sum_{j=1}^{n_i} \frac{K B^2}{\sigma_{\min}^2} \|x_{ij}\|^4 \|z_{ij}\|^2 \|\beta - \tilde{\beta}\|.
\end{aligned}$$
The term $\frac{\partial^2 h_i}{\partial \beta_k \partial u_i^\top}$  is bounded by
$$ 
\left\| \sum_{j=1}^{n_i} \frac{\partial \ell_{ij}}{\partial \beta_k \partial u_i^\top} \right\| \le C_B \sum_{j=1}^{n_i} \|x_{ij}\| \|z_{ij}\|.
$$
Now, assemble the bounds,
$$\begin{aligned}
\|\Delta_4\| &\le \| \frac{\partial \ell_{ij}}{\partial \beta_k \partial u_i^\top} \| \cdot \left\| [H_i(\Psi)]^{-1} - [H_i(\tilde{\Psi})]^{-1} \right\| \cdot \| \frac{\partial \ell_{ij}}{\partial \beta_k \partial u_i^\top} \| \\
&\le \| \frac{\partial \ell_{ij}}{\partial \beta_k \partial u_i^\top} \|^2 \cdot \frac{1}{\lambda_1^2} \cdot \| H_i(\Psi) - H_i(\tilde{\Psi}) \| \|\beta - \tilde{\beta}\|  \\
&\le \left( C_B \sum_{j=1}^{n_i} \|x_{ij}\| \|z_{ij}\| \right)^2 \cdot \frac{1}{\lambda_1^2} \cdot \left( \frac{K B^2}{\sigma_{\min}^2} \sum_{j=1}^{n_i} \|x_{ij}\|^4 \|z_{ij}\|^2 \right) \\
&= \frac{C_B}{\lambda_1^2} \left( \sum_{j=1}^{n_i} \|x_{ij}\| \|z_{ij}\| \right)^2 \left( \sum_{j=1}^{n_i} \|x_{ij}\|^4 \|z_{ij}\|^2\right) \|\beta - \tilde{\beta}\|.
\end{aligned}$$
Bound for $\Delta_2$ and $\Delta_3$: 
$$
\begin{aligned}
 \|\Delta_3\| \approx \|\Delta_2\| \le  \frac{C_b'}{\lambda_1} \left( \sum_{j=1}^{n_i} \|x_{ij}\|^5 \|z_{ij}\| \right) \left( \sum_{j=1}^{n_i} \|x_{ij}\| \|z_{ij}\| \right) \|\beta - \tilde{\beta}\|.
\end{aligned}
$$
According to Theorem 1, $\hat{\beta_k} \xrightarrow{p} \beta_k$; additionally, since the preceding term is bounded, $\Delta_1 + \Delta_2 + \Delta_3 + \Delta_4\xrightarrow{p} 0$.
Finally, regarding the first-order derivative part, the results hold because the mixed normal distribution density function is continuously twice differentiable, and the samples are independent. Further, by  Condition (B4) and Slutsky's theorem, we have
$$
\sqrt{N}(\hat\beta_k-\beta_k) \xRightarrow{d} \mathcal N\left(0, V_\beta \right),
$$
where 
$$
V_\beta = \bigg( \mathbb{E} \bigg[ \frac{\partial^2 \ell}{\partial \beta_k \partial \beta_k^\top} \bigg]_{\beta_0} \bigg)^{-1} 
\bigg( \mathbb{E} \bigg[ \frac{\partial \ell}{\partial \beta_k} \bigg(\frac{\partial \ell}{\partial \beta_k}\bigg)^\top \bigg]_{\beta_0} \bigg) 
\bigg( \mathbb{E} \bigg[ \frac{\partial^2 \ell}{\partial \beta_k \partial \beta_k^\top} \bigg]_{\beta_0} \bigg)^{-1}.
$$
\end{proof}

\end{document}